\documentclass[aps,prb,notitlepage,nofootinbib,superscriptaddress,twocolumn,longbibliography,10pt]{revtex4-2}

\usepackage[normalem]{ulem}
\usepackage{booktabs}

\newcommand{\sD}{{\mathsf D}}
\newcommand{\sU}{{\mathsf U}}
\newcommand{\sP}{{\mathsf P}}
\newcommand{\sG}{{\mathsf G}}

\newcommand{\CA}{\mathcal{A}}

\newcommand{\CC}{\mathcal{C}}

\newcommand{\CI}{\mathcal{I}}
\newcommand{\CH}{\mathcal{H}}

\newcommand{\BH}{\mathbb{H}}
\newcommand{\BF}{\mathbb{F}}
\newcommand{\tp}{\mathsf{F}_{\otimes}}
\newcommand{\chp}{\mathsf{F}_{\star}}

\newcommand{\beq}{\begin{eqnarray}}
\newcommand{\eeq}{\end{eqnarray}}
\newcommand{\beqq}{\begin{eqnarray*}}
\newcommand{\eeqq}{\end{eqnarray*}}

\newcommand{\be}{\begin{equation}}
\newcommand{\ee}{\end{equation}}
\newcommand{\barr}{\begin{array}}
\newcommand{\earr}{\end{array}}

\newcommand{\bra}[1]{{\langle #1|}}
\newcommand{\ket}[1]{{|#1 \rangle}}

\def\({\left(}
\def\){\right)}
\newcommand{\figref}[1]{Fig.~\ref{#1}}
\newcommand{\eqnref}[1]{Eq.~\eqref{#1}}
\newcommand{\secref}[1]{Sec.~\ref{#1}}
\newcommand{\appref}[1]{Appendix~\ref{#1}}

\def\parfig#1#2{
\parbox{#1\textwidth}
{\includegraphics[width=#1\textwidth]{#2}}
}

\usepackage{amsmath}
\usepackage{amsthm, amssymb}

\newtheorem{proposition}{Proposition}

\usepackage[colorlinks=true,citecolor=blue,linkcolor=blue]{hyperref}
\usepackage{graphicx}
\usepackage[ruled,vlined]{algorithm2e}
\usepackage{dcolumn}
\usepackage{bm}
\usepackage{color}
\usepackage{tabularx}
\usepackage{comment}
\usepackage{mathtools}
\usepackage{tikz-cd}
\usepackage{adjustbox}
\usepackage{dsfont}
\usepackage[english]{babel}
\makeatletter
\AtBeginDocument{%
  \@ifundefined{selectlanguage}{\newcommand{\selectlanguage}[1]{}}%
                                {\renewcommand{\selectlanguage}[1]{}}%
}
\makeatother
\usepackage[small]{caption}
\usepackage{subcaption}
\usepackage{multirow}
\begin{document}

\begin{titlepage}

\widetext

\title{Quantum Process Realization of LDPC Code Dualities and Product Constructions}

\author{Shuhan Zhang}
\author{Deepak Aryal}
\author{Yi-Zhuang You}
\affiliation{Department of Physics, University of California, San Diego, California 92093, USA}

\date{\today}

\begin{abstract}
We realize a broad class of code transformations, including Kramers–Wannier duality, tensor product, and check product, as quantum processes consisting of ancilla initialization, local unitaries, and projective measurements. Using ZX-calculus, we formulate these code transformations diagrammatically and develop a systematic algorithm for extracting their quantum-process realizations. Central to our framework is the observation that the physical content of a classical LDPC code is captured by the operator algebra associated with its Tanner graph, and that code transformations correspond to maps between such algebras. Kramers–Wannier duality then admits a natural interpretation as gauging, while tensor and check products correspond to coupled-layer constructions in which interlayer coupling and projection implement a quotient on stacked operator algebras. Together, these results establish a unified framework connecting code transformations, quantum circuits, and mappings between distinct quantum phases of matter.
\end{abstract}

\maketitle

\vspace{2mm}

\end{titlepage} 
\section{Introduction and motivation}
A central theme in the study of quantum phases of matter is to understand how transformations between distinct phases can be realized as physical processes. An increasingly fruitful perspective on this problem has emerged from quantum error correction, where different quantum phases can be associated with distinct code spaces. In this context, low-density parity-check (LDPC) codes \cite{gallagerLowdensityParitycheckCodes1962,sipserExpanderCodes2006, breuckmannQuantumLowDensityParityCheck2021,panteleevAsymptoticallyGoodQuantum2022, leverrierQuantumTannerCodes2022} provide a powerful and tractable framework for modeling quantum many-body systems, particularly those with geometrically non-local interactions or defined on non-Euclidean geometries—settings that are often difficult to access using conventional condensed-matter approaches.

This connection between codes and phases has been highlighted in recent work \cite{rakovszky_physics_2023}, which demonstrates that families of LDPC codes can encode a wide variety of gapped quantum phases. In particular, gauging a classical LDPC (cLDPC) code was shown to correspond to a generalized Kramers–Wannier (KW) duality transformation \cite{kramersStatisticsTwoDimensionalFerromagnet1941,wegnerDualityGeneralizedIsing1971}, mapping a trivial product-state phase to one exhibiting spontaneous symmetry breaking (SSB) \cite{rakovszky_physics_2023}. More generally, gauging classical codes with local redundancies yields quantum LDPC (qLDPC) codes associated with topologically ordered phases \cite{rakovszky_physics_2024, kitaevFaulttolerantQuantumComputation1997,levinStringnetCondensationPhysical2005}. These results illustrate that gauging transformations act as operational bridges between distinct quantum phases of matter.

Motivated by this perspective, it is natural to ask how gauging transformations can be physically realized. Recent works have shown that gauging can be implemented as explicit quantum processes involving ancilla qubits, local unitary circuits, and projective measurements. In particular, quantum process realizations of untwisted gauging in (1+1)-D have been demonstrated \cite{okada_non-invertible_2024,khan_quantum_2024,lu_realizing_2024}, with subsequent generalizations to higher dimensions and twisted gauging \cite{mana_kennedy-tasaki_2024,lu_realizing_2024}. These constructions enable explicit interpolations between distinct gapped quantum phases and provide concrete protocols for quantum state preparation \cite{tantivasadakarnHierarchyTopologicalOrder2023,tantivasadakarnShortestRouteNonAbelian2023,lootensLowdepthUnitaryQuantum2025,huPreparingCodeStates2025}. Notably, when applied to classical codes with local redundancies, these constructions prepare ground states of the corresponding gauged quantum codes, which can exhibit topological order and highly nonlocal entanglement structure. Beyond state preparation, such explicit quantum processes provide physical realizations of non-invertible symmetries \cite{seibergNoninvertibleSymmetriesLSMtype2024,shaoWhatsDoneCannot2024,chenAnalogsDeconfinedQuantum2025}, enabling the study of their fusion rules, anomalies, and connections to symmetry-protected topological (SPT) phases \cite{choiNoninvertibleDualityDefects2022,kaidiKramersWannierlikeDualityDefects2022,gorantla_tensor_2024}.

Moreover, LDPC codes admit a variety of product constructions, including tensor product, check product, and cubic product, which generate codes with controlled local redundancies. When combined with gauging, these constructions give rise to a rich landscape of quantum phases, including topological and fractonic phases \cite{vijay_fracton_2016,shirleyFoliatedFractonOrder2019,haahLocalStabilizerCodes2011,tanFractonModelsProduct2025, haahLocalStabilizerCodes2011}. Understanding how these product constructions can themselves be realized as quantum processes is therefore essential for developing a unified physical picture of code-based phase constructions. 

In this work, we develop a unified framework for realizing a broad class of LDPC code transformations as explicit quantum processes. At a technical level, we use the ZX-calculus \cite{coeckePicturingQuantumProcesses2017,duncanGraphtheoreticSimplificationQuantum2020,wetering_zx-calculus_2020,jeandelCompleteAxiomatisationZXCalculus2018} to represent these transformations diagrammatically. For KW duality, closely related ZX-diagrammatic representations of gauging and non-invertible duality on general bipartite graphs were developed in Ref.~\cite{gorantla_tensor_2024}, where the graphical calculus was also used to study the associated operator algebra and fusion rules. Building on this perspective, we formulate KW duality in the cLDPC setting and develop a systematic procedure for extracting explicit quantum processes from the corresponding ZX diagrams, consisting of ancilla initialization, local Clifford unitaries, and projective measurements. We then extend the same framework to tensor product, check product, and their generalizations. This yields explicit circuit-level realizations of these code transformations, as summarized in Tables~\ref{tab:dualities} and~\ref{tab:products}. For KW duality, our framework further reveals a direct correspondence between circuit resources---specifically, the numbers of ancilla qubits and projective measurements---and intrinsic code-theoretic quantities, namely redundancies and symmetries.

Underlying our framework is the observation that the physical content of a cLDPC code is captured by the operator algebra associated with its Tanner graph, and that code transformations are most naturally understood as maps between such operator algebras. The Hamiltonians and phase structures then follow as consequences of this algebraic structure, rather than serving as independent starting points. From this perspective, KW duality admits a natural interpretation as a gauging transformation \cite{fradkinPhaseDiagramsLattice1979,wegnerDualityGeneralizedIsing1971,kogutIntroductionLatticeGauge1979}, in which ancilla qubits and unitary couplings correspond to introducing dynamical gauge fields or symmetry defects, while projective measurements enforce Gauss law constraints and project onto the gauge-invariant sector. Tensor and check product constructions, on the other hand, correspond to coupled-layer constructions \cite{ma_fracton_2017}, where unitary operations implement inter-layer couplings between stacked operator algebras, and measurements effectively realize the strong-coupling limit that projects onto the product code's algebra. This perspective clarifies the distinct physical roles played by ancilla qubits, unitary couplings, and measurements across these constructions.

The remainder of the paper is organized as follows. In \secref{sec:cLDPC}, we review classical and quantum LDPC codes and introduce the operator algebra framework. In \secref{sec:Duality}, we formulate KW duality using ZX-diagrams, present a systematic circuit extraction procedure, and relate different realizations to the minimal-coupling and defect-condensation pictures of gauging. In \secref{sec:products}, we extend this framework to tensor and check product codes, interpreting them as coupled-layer models. Finally, in \secref{sec:pq-prod}, we introduce the general $(p,q)$-product, which unifies tensor and check products for more complicated product constructions.
\begin{table*}[!tb]
    \centering
    \caption{Classical and quantum LDPC codes: definitions and operator realizations.}
    \label{tab:code_realizations}
    \begin{tabular}{l c c l l}
        \toprule
        & \textbf{Chain Complex} & \textbf{Parity-Check Matrix} & \textbf{Bit operators} & \textbf{Check operators} \\
        \midrule
        \textbf{Classical code}
        & $C_1 \xrightarrow{\;\delta\;} C_0$
        & $\mathbf{\BH}$
        & $\forall i:\; \sigma_i^x$
        & $\forall a:\; \prod_{i\in \delta(a)} \sigma_i^z$ \\[4pt]
        \multirow{2}{*}{\textbf{Quantum code}}
        & \multirow{2}{*}{$C_2 \xrightarrow{\;\delta_x\;} C_1 \xrightarrow{\;\delta_z^\top\;} C_0$}
        & $\BH_z$
        & \multirow{2}{*}{$\forall i:\; \sigma_i^x,\; \sigma_i^z$}
        & $\forall a_z:\; \prod_{i\in \delta_z(a_z)} \sigma_i^z$ \\[2pt]
        & & $\BH_x$ & & $\forall a_x:\; \prod_{i\in \delta_x(a_x)} \sigma_i^x$ \\
        \bottomrule
    \end{tabular}
\end{table*}

\begin{table*}[t]
    \centering
    \caption{Dualities as unary morphisms between cLDPC codes $\mathcal{C} \to \mathcal{C}'$.}
    \label{tab:dualities}
    \renewcommand{\arraystretch}{1.6}
    \begin{tabular}{c c c c c c}
        \toprule
        \textbf{Operation} & \textbf{Code} & \textbf{Parity-Check Matrix} & \textbf{Operator } & \textbf{ZX Diagram} & \textbf{Circuit} \\
        \midrule
        \textbf{Kramers-Wannier} & $\mathcal{C} \to \mathcal{C}^\top$ & $\BH \to \BH^\top$ & $\mathsf{D}$ & \raisebox{-0.4\height}{\includegraphics[height=2.5cm]{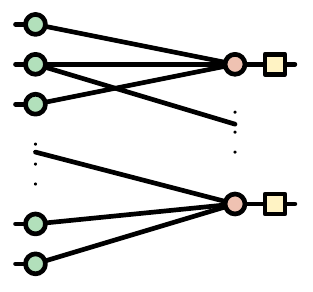}} & \raisebox{-0.4\height}{\includegraphics[height=2.5cm]{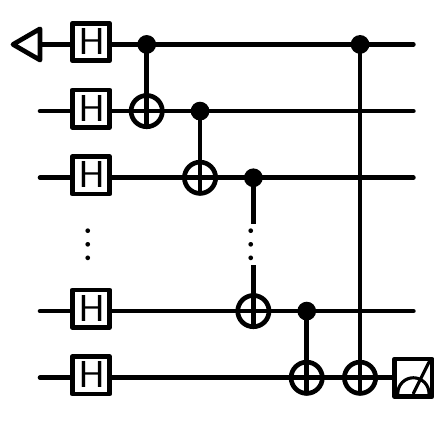}} \\
        \textbf{Orthogonal Complement } & $\mathcal{C} \to \mathcal{C}^\perp$ & $\BH \to \BH^\perp$ & $\mathsf{G}$ & \raisebox{-0.4\height}{\includegraphics[height=2.0cm]{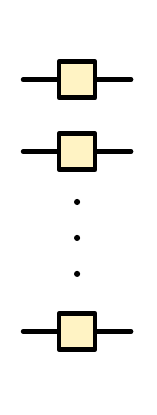}} & \raisebox{-0.4\height}{\includegraphics[height=2.0cm]{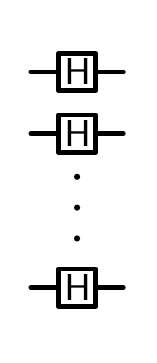}} \\
        \bottomrule
    \end{tabular}
\end{table*}

\begin{table*}[!tb]
    \centering
    \caption{Product constructions as binary morphisms between cLDPC codes $(\mathcal{C}_1, \mathcal{C}_2) \to \mathcal{C}_3$.}
    \label{tab:products}
    \renewcommand{\arraystretch}{1.6}
    \begin{tabular}{ c c c c c c }
        \toprule
        \textbf{Operation} & \textbf{Code} & \textbf{Parity-Check Matrix} & \textbf{Operator} & \textbf{ZX Diagram} & \textbf{Circuit} \\
        \midrule
        \textbf{Tensor Product} & $(\mathcal{C}_1, \mathcal{C}_2) \to \mathcal{C}_1 \otimes \mathcal{C}_2$ & $\BH_{\mathcal{C}_1}, \BH_{\mathcal{C}_2} \to \begin{bmatrix} \mathbb{I}_{n_2}\otimes \BH_{\mathcal{C}_1} \\ \BH_{\mathcal{C}_2}\otimes \mathbb{I}_{n_1} \end{bmatrix}$ & $\mathsf{F}_\otimes$ & \raisebox{-0.4\height}{\includegraphics[height=2.5cm]{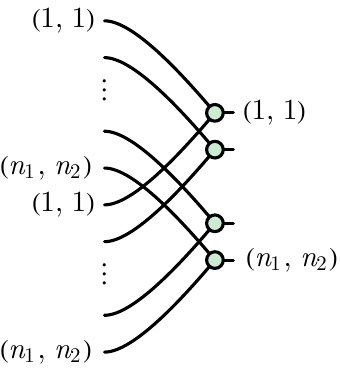}}  & \raisebox{-0.4\height}{\includegraphics[height=2.5cm]{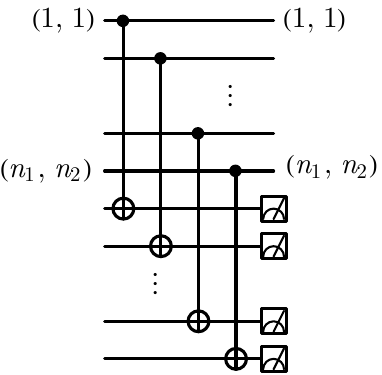}} \\
        \textbf{Check Product} & $(\mathcal{C}_1, \mathcal{C}_2) \to \mathcal{C}_1 * \mathcal{C}_2$ & $\BH_{\mathcal{C}_1}, \BH_{\mathcal{C}_2} \to \begin{bmatrix} \BH_{\mathcal{C}_1} \otimes \BH_{\mathcal{C}_2} \end{bmatrix}$ & $\mathsf{F}_*$ & \raisebox{-0.4\height}{\includegraphics[height=2.5cm]{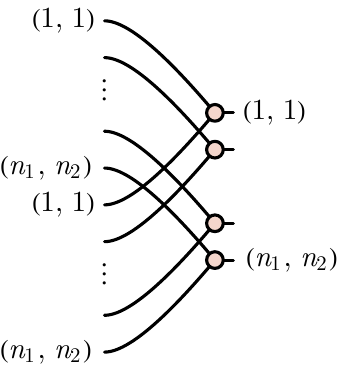}} & \raisebox{-0.4\height}{\includegraphics[height=2.5cm]{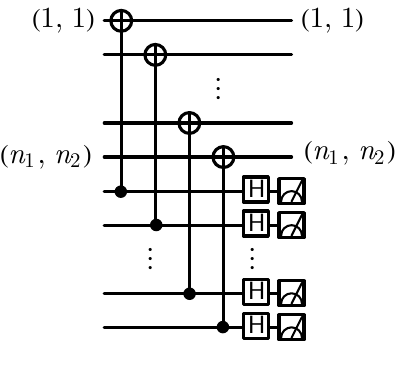}} \\
        \bottomrule
    \end{tabular}
\end{table*}

\section{Classical and quantum LDPC codes \label{sec:cLDPC}}
Table~\ref{tab:code_realizations} summarizes the correspondence between chain complexes \cite{bombinHomologicalErrorCorrection2007}, parity-check matrices, and operator realizations for classical and quantum LDPC codes.

\subsection{Definitions}
A cLDPC code $\mathcal{C}$ can be defined by a length-one chain complex of finite-dimensional binary vector spaces:
\begin{align}
    C_1 \xrightarrow{\;\delta\;} C_0 ,
\end{align}
where $C_0 = \mathbb{F}_2^{n}$ and $C_1 = \mathbb{F}_2^m$ are binary vector spaces of dimension $n$ and $m$, corresponding respectively to \emph{bits} and \emph{checks}. With a choice of bases for $C_0$ and $C_1$, the transpose map $\delta^\top: C_0 \to C_1$ is represented by an $m\times n$ binary matrix $\BH$, known as the \emph{parity-check matrix}. The \emph{code subspace} is defined as the kernel $\ker(\BH)$, consisting of all bit strings satisfying every parity check.

The parity-check matrix $\BH$ can also be viewed as a biadjacency matrix that defines a bipartite graph called the \emph{Tanner graph} $G = (V_b \cup V_c, E)$, where $V_b$ and $V_c$ denote the sets of $n$ bit nodes and $m$ check nodes, respectively. The edge set $E=\{(i,a) \in V_b \times V_c : \BH_{a,i}=1\}$ connects each bit $i$ to the check $a$ that acts on it.
We will use $\delta(a)=\{i\in V_b : \BH_{a,i}=1\}$ to denote the support of check $a$ in $C_0$, i.e., the set of bits connected to $a$ in the Tanner graph, and similarly $\delta^\top(i)=\{a\in V_c : \BH_{a,i}=1\}$ for the set of checks acting on bit $i$.

We embed the cLDPC code $\mathcal{C}$ into a quantum Hilbert space $\mathcal{H}_\sigma = (\mathbb{C}^2)^{\otimes n}$ by promoting each classical bit to a qubit. 
For each check $a$, we define an associated check operator, with $\sigma_i^z=\ket{0}\bra{0}_i-\ket{1}\bra{1}_i$ being the Pauli $Z$-operator acting on qubit $i$,
\begin{align}
    C_a=\prod_{i \in \delta(a)} \sigma_i^z,
\end{align}
such that the quantum code subspace is the common eigenspace of $C_a=1$ for all $a$. $\ket{\mathbf{0}}:=\prod_{i} \ket{0}_i$ is always a state in the code subspace. For any codeword $\lambda \in\ker(\delta^\top)$, we define an associated logical operator, with $\sigma_i^x=\ket{0}\bra{1}_i+\ket{1}\bra{0}_i$ being the Pauli $X$-operator acting on qubit $i$,
\begin{align}
    L_{\lambda} = \prod_{i\in \lambda} \sigma_i^x,
\end{align}
which maps $\ket{\mathbf{0}}$ to the corresponding codeword state $\ket{\lambda}=L_{\lambda}\ket{\mathbf{0}}$. The quantum code subspace $\mathrm{span}\{\ket{\lambda}\}$ is spanned by these codeword states.

On the other hand, the kernel of the map $\delta$ defines the space of \textit{redundancies} of the code. For any redundancy $R \in\ker(\delta)$, the corresponding product of check operators is trivial:
\begin{align}
    \prod_{a\in R} C_a = \openone .
\end{align}

A qLDPC code \cite{breuckmannQuantumLowDensityParityCheck2021,calderbankGoodQuantumErrorCorrecting1996,steaneErrorCorrectingCodes1996} is defined by a length-two chain complex of binary vector spaces
\begin{align}
    C_2 \xrightarrow{\delta_2=\delta_x} C_1 \xrightarrow{\delta_1=\delta_z^\top} C_0 ,
\end{align}
satisfying $\delta_1\circ\delta_2=0$. Physically, $C_1$ corresponds to qubits, and $C_0$ and $C_2$ correspond to $Z$-type and $X$-type checks, respectively. The condition $\delta_1\circ\delta_2=0$ ensures that all $Z$- and $X$-type stabilizers commute, so that the chain complex defines a valid quantum stabilizer code.

\subsection{Operator algebra and Hamiltonian for cLDPC codes\label{sec:op-alg}}

The physical content of a cLDPC code $\CC$ is most naturally captured not by a specific Hamiltonian, but by the \emph{operator algebra} associated with its Tanner graph. Given a cLDPC code with Tanner graph $G = (V_b \cup V_c, E)$, we associate a qubit to each bit node. 

Each node of the Tanner graph carries a single generating operator: $\sigma_i^x$ for bit nodes and $C_a = \prod_{i \in \delta(a)} \sigma_i^z$ for check nodes. We define the operator algebra of the code, $\CA_\CC$, as the subalgebra of $\mathrm{End}(\CH_\sigma)$ generated by the single-site operators $\{\sigma_i^x\}$ and the check operators $\{C_a\}$
\begin{align}
   \CA_\CC = \left\langle \sigma_i^x \;\big|\; i \in V_b \;;\; C_a \;\big|\; a \in V_c \right\rangle.
\end{align}

The most general Hamiltonian linear in the generators  of $\CA_{\CC}$ takes the form
\begin{align}\label{eq:cH}
    H_\CC = -J \sum_{a \in V_c} C_a - h \sum_{i \in V_b} \sigma_i^x,
\end{align}
where $J$ and $h$ are coupling constants associated with check and bit nodes, respectively. Tuning the ratio $h/J$ interpolates between the code-constrained phase at $h \ll J$, whose ground-state subspace coincides with the quantum code subspace, and a trivial product-state phase at $h \gg J$.

This formulation emphasizes that the Hamiltonian is entirely determined by the Tanner graph structure: the code constraints define the $Z$-type terms, while the bit-node operators provide the conjugate $X$-type terms. In particular, the transverse-field term $\sum_i \sigma_i^x$ is not an external deformation but an intrinsic part of the operator algebra. 

More generally, transformations between codes—such as KW duality and product constructions—can be understood as maps between operator algebras, from which the corresponding Hamiltonian mappings follow as a consequence. This perspective will be central to the physical interpretation of the constructions in Secs.~\ref{sec:Duality} and~\ref{sec:products}.

\section{Dualities and their quantum process realizations \label{sec:Duality}}
% Given a classical code $\mathcal{C}$ defined by a linear map $\delta : C_0\to C_1$,  the transpose code $\mathcal{C}^\top$ is defined by the linear map $\delta^\top: C_1 \to C_0$. Let us note that to define $\delta^\top$ we have to specify a basis for $C_0$ and $C_1$. The parity check matrix of transpose code is $\BH^\top$. So, the Tanner graph of $\delta^\top$ has the roles of qubit and checks switched in comparison to the Tanner graph of $\delta$. The Hamiltonian of the transpose code is given as:
% \begin{align*}
%     H_{\delta^\top} = - \sum_{i} \prod_{a\in \delta(i)} \tau^z_i - \sum_a \tau^x_a
% \end{align*}
% \par
The generalized KW duality exchanges bits and checks of a classical code $\CC$. By embedding $\CC$ and its KW dual code $\CC^{\top}$ into the quantum Hilbert spaces $\mathcal{H}_\sigma$ and $\mathcal{H}_\tau$, respectively, the KW transformation can be described by a non-invertible operator
\[
\sD : \mathcal{H}_\sigma \rightarrow \mathcal{H}_\tau
\]
that maps:
\begin{align}
    \sD \;\sigma_i^x &= \prod_{a\in\delta^\top(i)}\tau^z_a\; \sD , \label{eq:D1} \\
    \mathsf{D}\; \prod_{i \in \delta(a)} \sigma_i^z &= \tau_a^x \;\mathsf{D} . \label{eq:D2}
\end{align}
\par
KW duality for classical LDPC codes was understood as gauging in the work of Rakovszky and Khemani \cite{rakovszky_physics_2023}. A closely related graphical realization was subsequently developed by
Gorantla, Shao, and Tantivasadakarn \cite{gorantla_tensor_2024}, who constructed ZX-diagram representations of gauging and non-invertible duality maps for generalized Ising models
defined on arbitrary bipartite graphs. Here we focus on the quantum-process realization of this transformation in the cLDPC setting. Starting from its Tanner-graph ZX representation, we develop a systematic circuit-extraction procedure and relate the resulting ancilla and measurement resources directly to the symmetries and redundancies of the code. We then compare different process realizations with the defect-condensation and minimal-coupling pictures of gauging.

Another duality that exchanges the logicals and checks of a cLDPC code is the orthogonal-complement duality. By embedding $\CC$ and its orthogonal-complement dual code $\CC^{\perp}$ into the quantum Hilbert spaces $\mathcal{H}_\sigma$ and $\mathcal{H}_\gamma$, respectively, the orthogonal-complement dual transformation can be described by the operator
\[
\sG : \mathcal{H}_\sigma \rightarrow \mathcal{H}_\gamma
\]
that maps:
\begin{align}
    \sG \prod_{i\in \lambda}\sigma_i^x &= \prod_{i\in \lambda}\gamma_i^z \sG , \\
    \sG \prod_{i \in \delta(a)} \sigma_i^z &= \prod_{i \in \delta(a)} \gamma_i^x \sG .
\end{align}
It is straightforward to verify that $\sG$ can be implemented as a tensor product of Hadamard gates on all bits, $\sG=H^{\otimes n}$. In contrast to the KW duality operator $\sD$, the orthogonal-complement duality $\sG$ is invertible, reflecting the fact that it maps between Hilbert spaces of equal dimension without any projection. This distinction makes the quantum process realization of KW duality considerably richer: since
$\sD$ is non-invertible, its implementation necessarily involves ancilla qubits and projective measurements, whose structure we now make precise using the ZX-diagram framework.
\subsection{ZX diagram representation}
The ZX-calculus \cite{coecke_interacting_2011,wetering_zx-calculus_2020,luStrangeCorrelatorString2025} is a graphical language for representing quantum processes as tensor networks built from two types of vertices: Z-spiders and X-spiders, which correspond to certain families of unitary and projection operations. Its graphical rewrite rules preserve the underlying linear maps and enable rigorous derivations of circuit identities. As we show below, the ZX formalism is particularly well suited for our purposes: it provides a natural representation of the KW duality operator 
$\sD$ from which the required ancilla, unitaries, and measurements can be systematically extracted.

\begin{figure}[htbp]
    \centering
    \includegraphics[width=1\linewidth]{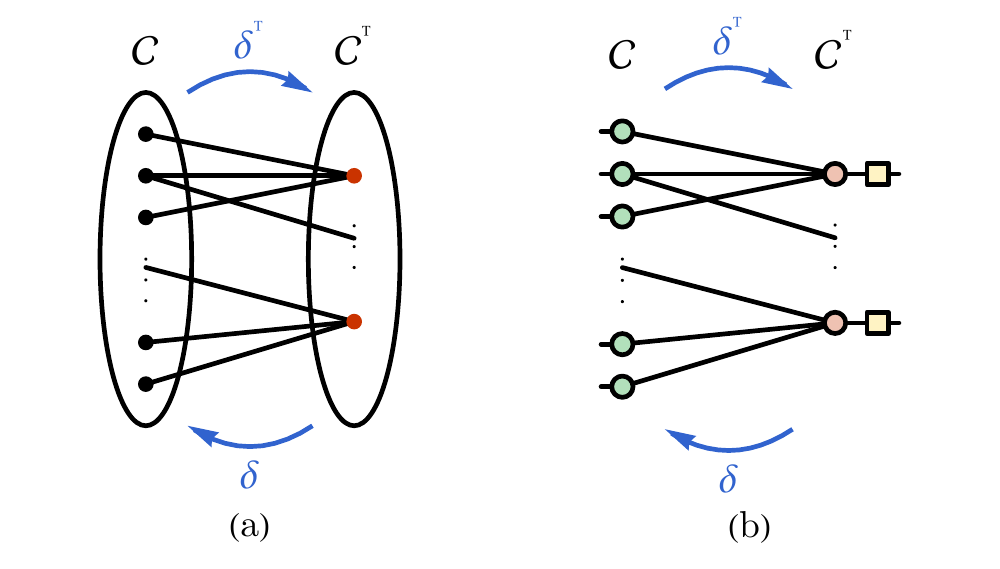}
    \caption{(a) Tanner graph representation of a cLDPC code. (b) ZX-diagram representation of KW duality operator $\sD$.}
    \label{fig:tanner}
\end{figure}

For a general cLDPC code, the KW duality map admits a natural ZX-diagram representation constructed directly from the Tanner graph. Closely related ZX-diagram representations were developed in Ref.~\cite{gorantla_tensor_2024} for generalized transverse-field Ising models defined on arbitrary bipartite graphs, where the biadjacency matrix is equivalently the parity-check matrix of the associated classical code. In that construction, the corresponding ZX map implements gauging and exchanges the two parts of the bipartite graph.

Adopting this graphical representation in the cLDPC setting, we replace the bit nodes by Z-spiders and the check nodes by X-spiders, as shown in
\figref{fig:tanner}. The Tanner graph encodes the bit-check incidence structure through the maps $\delta$ and $\delta^\top$, while the ZX-calculus realizes the corresponding operator map. The Hadamard nodes account for the change of basis between Pauli $Z$ and $X$ operators. The resulting diagram represents $\sD$ and implements the maps in Eqs.~\eqref{eq:D1} and~\eqref{eq:D2}.

Our focus below is on the quantum-process realization of this map: starting from the Tanner-graph ZX diagram, we develop a systematic procedure for extracting a circuit composed of ancilla initialization, local Clifford gates, and projective measurements, and characterize its resource requirements in terms of the symmetries and redundancies of the code.

\subsection{ZX diagram to quantum process 
\label{sec:circuit-extraction}}
Given a ZX diagram obtained from the Tanner graph, as in \figref{fig:tanner}(b), we now describe a concrete procedure for constructing the quantum process consisting of ancilla attachment, local unitary
operations, and projective measurements. The circuit is extracted by performing Gaussian elimination on the parity-check matrix: each row operation corresponds to a \texttt{CNOT} gate in the resulting circuit.

Our construction follows Algorithm \ref{alg:process-extraction} adapted from Backens et al.~\cite{backens_there_2021} for extracting quantum circuits from ZX diagrams. The key idea is to perform
Gaussian elimination on the biadjacency matrix associated with the underlying bipartite graph of the ZX diagram. As illustrated in \figref{fig:cnot_extraction}, each elementary row operation during Gaussian elimination corresponds to the extraction of a \texttt{CNOT} gate in the resulting circuit. Here we adopt a simplified notation for visual clarity:
    \parfig{0.15}{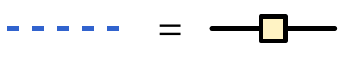}.

\begin{figure}[ht]
    \centering
    \includegraphics[width=0.7\linewidth]{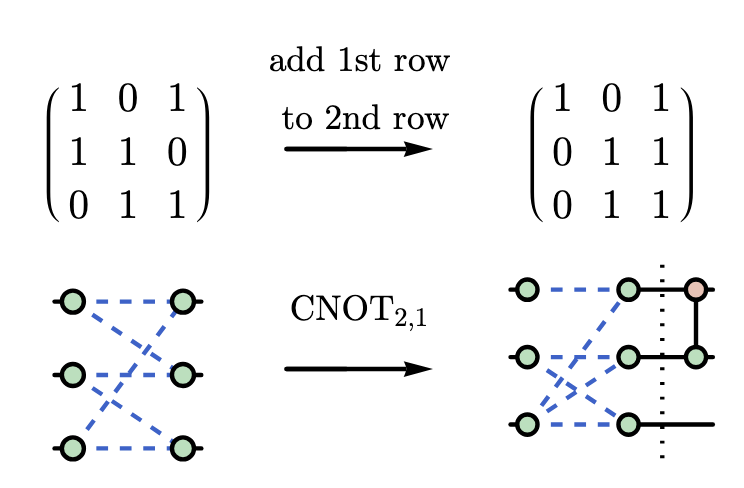}
    \caption{Basic move in circuit extraction: for row operation $R_2 \mapsto R_1 +R_2$ of the biadjacency matrix of the ZX diagram, we extract a \texttt{CNOT(control = $q_2$, target = $q_1$)} gate.}
    \label{fig:cnot_extraction}
\end{figure}

\begin{algorithm}[ht]
\caption{Extract Quantum Process from $\BH$}
\label{alg:process-extraction}
\begin{flushleft}
\textbf{Input:} Parity-check matrix $\BH \in \BF_2^{m \times n}$\\
\textbf{Output:} Quantum process realization of $\sD$
\end{flushleft}
\vspace{-0.5em}
\begin{enumerate}
    \item Perform Gaussian elimination on $\BH$:
    \begin{enumerate}
        \item For each row operation $R_i \leftarrow R_i + R_j$:
        \begin{itemize}
            \item Emit \texttt{CNOT(control = $q_i$, target = $q_j$)} on the output side
            \item Update the ZX-diagram connectivity to reflect the new row-reduced matrix $\BH'$
        \end{itemize}
    \end{enumerate}
    \item Perform Gaussian elimination on $\BH'^\top$:
    \begin{enumerate}
        \item For each row operation $R_i \leftarrow R_i + R_j$:
        \begin{itemize}
            \item Emit \texttt{CNOT(control = $q_i$, target = $q_j$)} on the input side
            \item Update the ZX-diagram connectivity to reflect the new row-reduced matrix
        \end{itemize}
    \end{enumerate}
\end{enumerate}
\end{algorithm}
 
An explicit example applying this procedure to the 1D Ising model is presented in \appref{app:I1D-extract}. Upon completion of the algorithm, the original ZX diagram is converted into a quantum process composed of three elementary ingredients: ancilla preparation, \texttt{CNOT} gates, and projective measurements. The resulting circuit realization is not unique, as it depends on the specific sequence of row operations chosen during Gaussian elimination.

We now show that any quantum process realizing $\sD$ necessarily requires at least $k^\top$ ancilla qubits and $k$ projective measurements, where $k$ and $k^\top$ are the dimensions of the kernels of $\BH$ and $\BH^\top$, respectively.
\begin{proposition}[Minimal resources for quantum process realizing $\sD$]
Let $\mathbb{H} \in \mathbb{F}_2^{m \times n}$ be the parity-check matrix of a cLDPC code $\mathcal{C}$, and let $\sD$ denote the generalized KW duality map. Then the minimal number of ancilla qubits and projective measurements required to realize $\sD$ as a quantum process are $k^\top = m - \operatorname{rank}(\BH)$ and $k = n - \operatorname{rank}(\BH^\top)$ respectively (the number of redundancies and symmetries of $\CC$).
\end{proposition}
\noindent
\textit{Proof:} After step~1 of Algorithm~\ref{alg:process-extraction}, Gaussian elimination on $\mathbb{H}$ produces
\[
m-\operatorname{rank}(\mathbb{H})
= m-(n-\dim\ker(\mathbb{H}))
= m-(n-k)
= k^\top
\]
zero rows in the biadjacency matrix. Each such zero row corresponds to a dangling output leg in the ZX diagram and therefore requires the attachment of an ancilla qubit. Hence at least $k^\top$ ancilla qubits are necessary.

Similarly, after step~2, Gaussian elimination on the transposed matrix $\mathbb{H}^\top$ produces
\[
n-\operatorname{rank}(\mathbb{H}^\top)
= n-(m-\dim\ker(\mathbb{H}^\top))
= n-(m-k^\top)
= k
\]
zero rows, which correspond to dangling input legs and thus require $k$ projective measurements. This establishes the lower bound.

Since the circuit extracted by Algorithm~\ref{alg:process-extraction} achieves
these numbers explicitly, the bounds are tight.
\hfill$\square$

We can also prove the following proposition which will be useful in the gauging construction of Sec.~\ref{sec:gauging}.
\begin{proposition}[Unitary ZX diagrams]\label{prop:uni}
If the parity-check matrix $\BH$ is square and full-rank, then the corresponding ZX diagram represents a unitary map.
\end{proposition}
\noindent 
\noindent\textit{Proof.}
 If $\BH$ is square, then $n=m$. Full rank then implies $\ker(\BH) = \ker(\BH^\top) = 0$, so the circuit extraction algorithm introduces no ancilla qubits and no projective measurements. The resulting quantum process consists solely of unitaries.
\hfill$\square$

\subsection{Physical interpretation as gauging} \label{sec:gauging}

The generalized KW map $\sD$ is a non-invertible operator, but it can always be realized as a unitary through embedding into a larger Hilbert space. Concretely, $\sD$ admits a decomposition of the form
\begin{align}
    \sD = \sP_M \, \sU_{A,M} \, \sP_A ,
\end{align}
where $\sU_{A,M}$ is a unitary operator acting on an enlarged Hilbert space, while $\sP_A$ and $\sP_M$ are projection maps corresponding to ancilla initialization and projective measurement in specified bases. Such a decomposition is not unique. Different choices of ancilla and measurement spaces lead to distinct quantum processes realizing the same non-invertible map $\sD$, and these different realizations admit natural and complementary physical interpretations as gauging procedures.

The same map $\sD$ can thus be realized with different numbers of ancillas and measurements; we describe two canonical choices that correspond to the defect-condensation and minimal-coupling pictures of gauging.

\subsubsection{Defect Condensation}
Applying the Gaussian-elimination-based circuit extraction of \secref{sec:circuit-extraction} yields a realization of the form
\begin{align}
    \sD = \sP_\mu \, \sU_{\mu,\eta} \, \sP_\eta ,
\end{align}
in which $k^\top$ ancilla qubits are introduced and $k$ projective measurements are performed. The corresponding ZX-diagram is shown in
\figref{fig:UE1}. 

Operationally, adding $k^\top$ ancilla legs and $k$ measurement legs corresponds to appending $k^\top$ rows and $k$ columns to the original parity-check matrix $\mathbb{H}$, producing a square matrix $\mathbb{H}'$ satisfying
\begin{align}
    \ker(\mathbb{H}') = \ker(\mathbb{H}'^{\top}) = 0 .
\end{align}
As shown in Proposition~\ref{prop:uni}, this guarantees that the resulting ZX diagram represents a fully unitary map $\sU$. This construction admits a natural interpretation as a lifting of the non-invertible map $\sD$ to a unitary transformation between operators in two enlarged Hilbert spaces of equal dimension $\sU:\CH_{\sigma,\eta} \to \CH_{\tau, \mu}$. As shown in \figref{fig:UE1}, the ancilla attachment extends the $2^n$-dimensional Hilbert space $\CH_\sigma$ of the original code $\CC$ to an enlarged space $\CH_{\sigma,\eta}$ of dimension $2^{n+k^\top}$, which matches the dimension of the enlarged output space $\CH_{\tau,\mu}$ of dimension $2^{m+k}$. 

Using Algorithm~\ref{alg:process-extraction}, the corresponding circuit can be written explicitly as
\begin{align}
    \sD = \sP_\mu  \prod_{\alpha=1}^{k}\prod_{i\in \lambda_\alpha} \texttt{CNOT}_{\alpha, i}\;\prod_{\beta=1}^{k^\top}\prod_{a\in R_\beta} \texttt{CNOT}_{\beta, a} \; \sP_\eta,
\end{align}
where $\{\lambda_\alpha\}_{\alpha=1}^{k}$ are the symmetries of $\CC$, and $\{R_\beta\}_{\beta=1}^{k^\top}$ are the redundancies of $\CC$, or equivalently, logical operators of the transpose code $\CC^\top$. Following the convention of \figref{fig:cnot_extraction}, the first subscript denotes the control and the second the target. This structure reflects the fact that the row operations required to eliminate rows of $\mathbb{H}$ correspond precisely to redundancies of the check constraints (i.e., logical operators of $\CC^\top$), while the row operations on $\mathbb{H}^\top$ correspond to logical operators of $\CC$ itself.
\begin{figure}[htbp]
    \centering
    \includegraphics[width=1\linewidth]{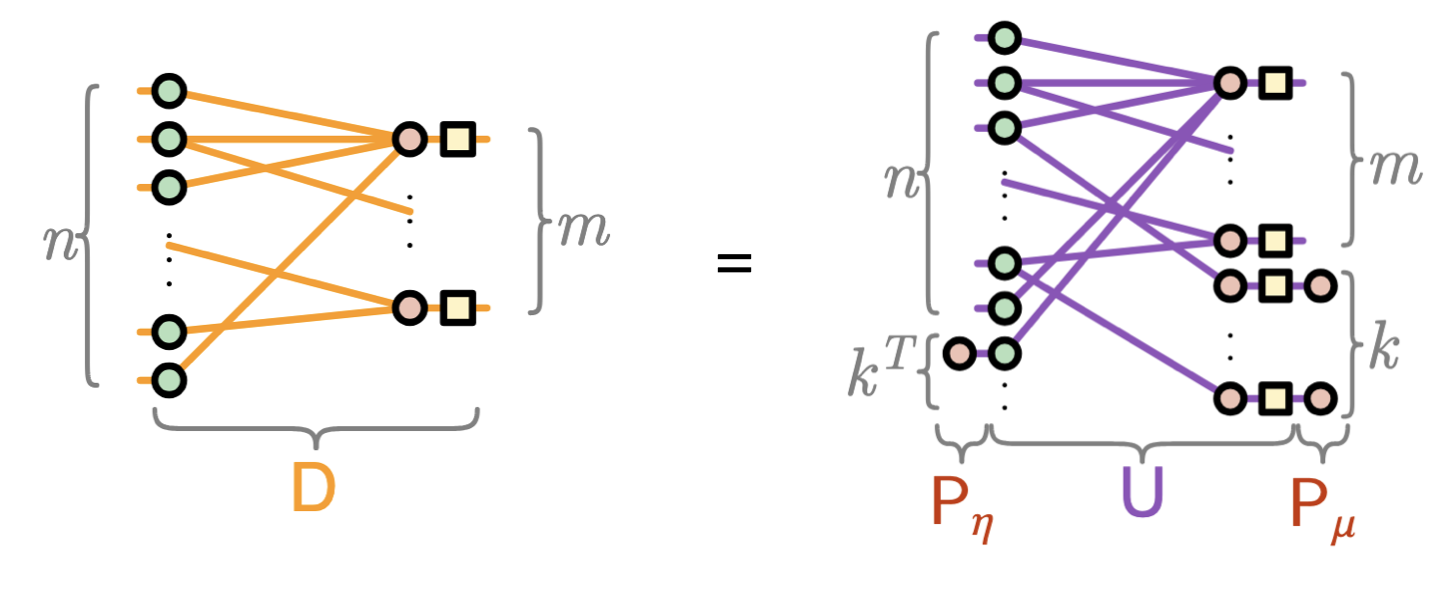}
    \caption{Defect condensation picture: ZX-diagram representation of the quantum process realization with $k^\top$ ancillae and $k$ measurements $\sD = \sP_\mu \, \sU_{\mu,\eta} \, \sP_\eta$.}
    \label{fig:UE1}
\end{figure}

Physically, this circuit realization corresponds to gauging via defect condensation \cite{gorantla_tensor_2024}. Here, the code-theoretic structure allows us to identify the required ancilla resources explicitly with the $k^\top$ redundancies of the parity-check matrix, while the $k$ measurements are associated with the symmetries of the original code.

\subsubsection{Minimal Coupling}
An alternative realization of $\sD$ is obtained by choosing a maximal embedding, leading to a decomposition
\begin{align}
    \sD = \sP_\gamma \, \sU'_{\gamma,\kappa} \, \sP_\kappa ,
\end{align}
in which $m$ ancilla qubits and $n$ projective measurements are introduced. The corresponding ZX-diagram is shown in \figref{fig:UE2}.

\begin{figure}[htbp]
    \centering
    \includegraphics[width=1\linewidth]{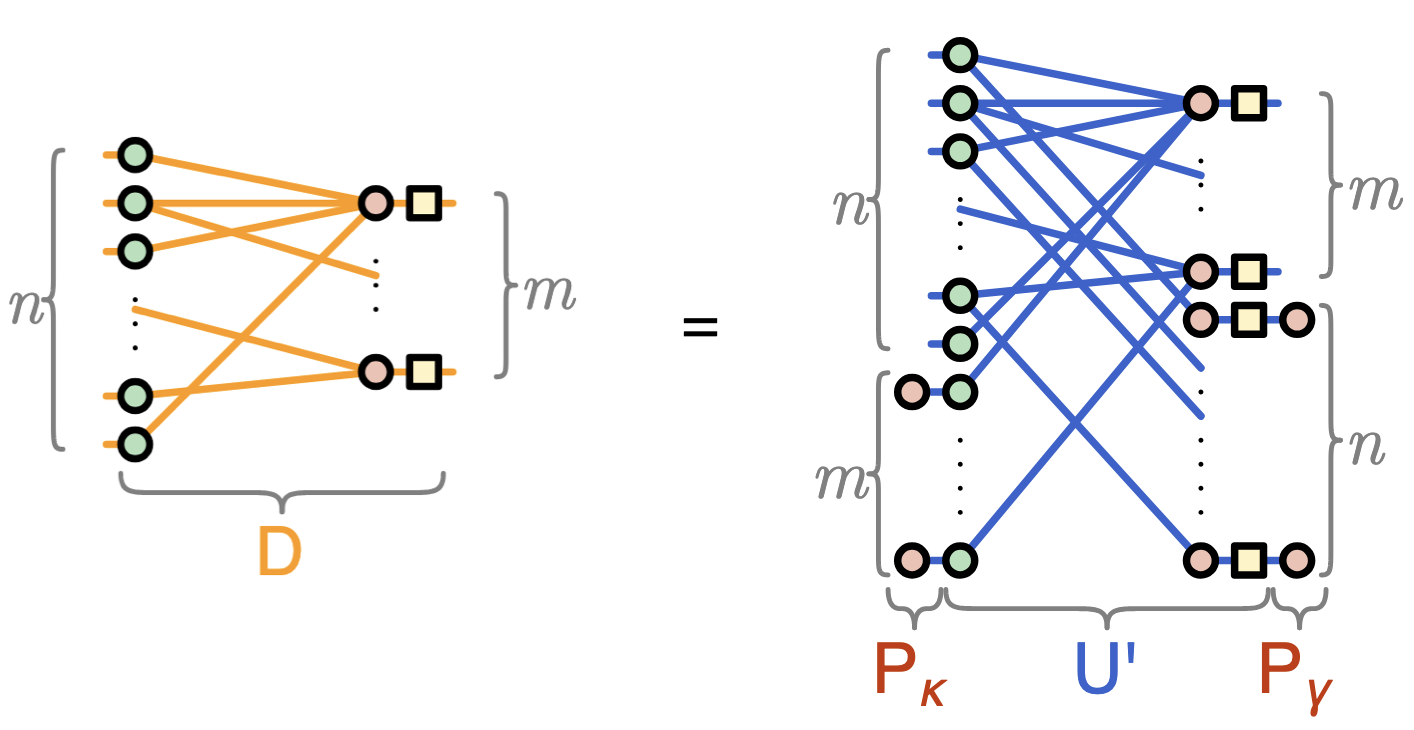}
    \caption{Minimal coupling picture: ZX-diagram representation of the quantum process realization with $m$ ancilla qubits and $n$ measurements $\mathsf{D} = \sP_\gamma \, \sU'_{\gamma,\kappa} \, \sP_\kappa$.}
    \label{fig:UE2}
\end{figure}

Applying Algorithm~\ref{alg:process-extraction} yields the circuit
\begin{align}
    \sD = \sP_\gamma \prod_{i=1}^{n}\prod_{a\in \delta^\top(i)} \texttt{CNOT}_{i, \tilde{a}} \;\prod_{a=1}^{m}\prod_{i\in \delta(a)} \texttt{CNOT}_{a, \tilde{i}}\;  \sP_\kappa
\end{align}
where tildes distinguish the ancilla qubits from the original code qubits. Each of the $m$ ancilla qubits $\tilde{a}$ corresponds to a check node, and each of the $n$ measurements corresponds to a bit node.

This realization corresponds to \emph{gauging via minimal coupling} \cite{barkeshli_symmetry_2019,harlow_symmetries_2019}. In this picture, the ancilla qubits represent dynamical local gauge fields, while the projective measurements impose Gauss law constraints that project onto the gauge-invariant
subspace.

A detailed derivation demonstrating the precise correspondence between these quantum processes and the two gauging procedures is provided in Appendix~\ref{app:gauge}.

\subsection{Phase interpretation} \label{sec:stateprep}
The KW duality map $\sD$ can also be understood as a transformation between operator algebras and, consequently, between quantum phases of matter. As shown in Eqs.~(\ref{eq:D1}) and~(\ref{eq:D2}), $\sD$ maps the operator algebra $\CA_\CC$ of the original code to the operator algebra $\CA_{\CC^\top}$ of the transpose code. Since the Hamiltonians $H_\CC$ and $H_{\CC^\top}$ are determined by their respective operator algebras via Eq.~(\ref{eq:cH}), the duality $\sD$ induces a mapping between the corresponding quantum phases. The quantum process constructed above therefore provides not only an implementation of the abstract operator-algebra map, but also a concrete and operational protocol for preparing quantum states \cite{tantivasadakarnHierarchyTopologicalOrder2023, tantivasadakarnShortestRouteNonAbelian2023} in nontrivial phases.

For cLDPC codes without local redundancies, $\sD$ maps the trivial product-state ground state of $H_\CC$ at $h \gg J$ to a ground state of $H_{\CC^\top}$ at $h \ll J$ exhibiting SSB. The ground-state manifold is in general degenerate, spanned by distinct symmetry-breaking states.

Importantly, the ancilla qubits introduced in our construction play a crucial role in selecting a particular ground state from this degenerate manifold. In the defect condensation realization, $k^\top$ ancilla qubits are initialized in the $\ket{+}^{\otimes k^\top}$ state, enforcing the condition that all global symmetries—equivalently, the logical operators of the dual code $\CC^\top$—have eigenvalue $+1$. This preparation selects the maximally symmetric ground state within the SSB phase. More generally, by adjusting the initialization basis and relative phases of the ancilla qubits, one can prepare other symmetry-broken ground states, thereby achieving explicit control over the symmetry sector of the resulting state. In fact, for any ground state $\ket{\psi}$ in the SSB manifold, there exists an ancilla initialization such that the quantum process maps a trivial product state to $\ket{\psi}$. An explicit example of this state-preparation procedure for the one-dimensional Ising model is presented in Appendix~\ref{app:stateprep}.

For cLDPC codes with local redundancies, the same process prepares states in a topological order phase. In this
case, the ancilla qubits select a topological sector within the degenerate ground-state manifold of the gauged code, analogous to their role in selecting a symmetry sector in the SSB case.

%For example, it takes the paramagnetic phase (trivial phase) of a 1d quantum Ising model to the ferromagnetic phase (spontaneously symmetry broken phase). When the input classical code has local redundancies, KW duality leads to a nontrivial phase of matter, like in the case of 2d Ising model KW duality maps it to $\mathbb{Z}_2$ gauge theory (toric code topological order).

% Explicit realization of the map in terms of ZX diagrams, and furthermore our algorithm to extract explicit circuit from it, gives a concrete method to prepare states from non-trivial phases of matter from trivial states.

\section{Product constructions and coupled layer processes \label{sec:products}}
In the previous section, we showed how the KW transformation of a single cLDPC code can be realized as an explicit quantum process and interpreted as a mapping between operator algebras and quantum phases. We now extend this framework to product constructions built from multiple cLDPC codes. At the operator-algebra level, these constructions take the operator algebras of the input codes, stack and couple them, and project onto a quotient algebra that generates the product code. From a physical perspective, this corresponds to a coupled-layer construction \cite{ma_fracton_2017,vijayIsotropicLayerConstruction2017,tanFractonModelsProduct2025}, in which stacks of simpler code Hamiltonians are entangled and projected to generate more complex check structures. We show that tensor products and check products can be represented and analyzed using the same ZX-diagrammatic language and quantum-process formalism, and that their physical content is most transparently understood through the operator-algebra framework of Sec.~\ref{sec:op-alg}.

Let $\mathcal{C}_1$ and $\mathcal{C}_2$ be two classical codes with parity check matrices $[\BH_{\mathcal{C}_1}]_{m_1\times n_1}$ and $[\BH_{\mathcal{C}_2}]_{m_2\times n_2}$, boundary maps $\delta^1$ and $\delta^2$, sets of checks $\{C^1_a | a \in A\}$ and $\{C^2_b | b \in B\}$, and embedding Hilbert spaces $\CH_\alpha$ and $\CH_\beta$. Following the operator-algebra framework of Sec.~\ref{sec:op-alg}, we associate to each code its canonical Hamiltonian built from the full operator algebra $\CA_{\CC_s}$. This includes both the check terms (with coupling $J$ set to 1) and the bit-node terms (with coupling $h_s$), yielding
\begin{align}
    H_1&=- \sum_{a \in A}C^1_a - h_1 \sum_i \alpha_i^x,\\
    H_2&=- \sum_{b \in B}C_b^2 - h_2\sum_j \beta_j^x,
\end{align}
where $C^1_a=\prod_{i \in \delta^1(a)}\alpha_i^z$, and $C^2_b=\prod_{j \in \delta^2(b)}\beta_j^z$. These Hamiltonians interpolate between the code-constrained phase at small $h_s$ and a trivial product-state phase at large $h_s$.
\subsection{Tensor product code}
We form parity-check matrices by vertical concatenation of blocks (row stacking), denoted with $\left[\begin{smallmatrix} \BH_{S_1} \\ \BH_{S_2} \end{smallmatrix}\right]$. The tensor product code $\mathcal{C}_1\otimes \mathcal{C}_2$ has parity check matrix
\begin{align*}
    \BH_{\mathcal{C}_1\otimes \mathcal{C}_2} = \begin{bmatrix}
        \mathbb{I}_{n_2}\otimes \BH_{\mathcal{C}_1} \\
        \BH_{\mathcal{C}_2}\otimes \mathbb{I}_{n_1}
    \end{bmatrix}_{(n_1m_2+n_2m_1)\times n_1n_2}.
\end{align*}
The codewords of the tensor product code in $\mathbb{F}_2^{n_1n_2}$ can be cast into $n_1\times n_2$ dimensional binary matrices. In this representation, a binary $n_1\times n_2$ matrix is a codeword of the tensor product code if and only if every row is a codeword of $\mathcal{C}_1$ and every column is a codeword of $\mathcal{C}_2$. 

Let $\CH_\sigma$ be the embedding Hilbert space for $\CC_1 \otimes \CC_2$. The tensor product code has two types of check operators: row checks $C^1_{a,j}$ inherited from $\CC_1$ acting on each row $j$, and column checks $C^2_{i,b}$ inherited from $\CC_2$ acting on each column $i$
\begin{align} \label{eq:tensorchecks}
    C^1_{a,j}=\prod_{i \in \delta^1(a)}\sigma_{ij}^z, \quad
    C^2_{i,b}=\prod_{j \in \delta^2(b)}\sigma_{ij}^z .
\end{align}
The logical operators of $\CC_1 \otimes \CC_2$ can be obtained by taking the tensor product of the logical operators of the two codes $\{L^{1}_\lambda|\lambda \in \Lambda\}, \{L^{2}_\nu | \nu \in N \}$
\begin{align}
    L_{\lambda, \nu} = \prod_{i\in \lambda, j \in \nu} \sigma^x_{ij}.
\end{align}
\subsubsection{ZX realization as quantum process}
The tensor product can be realized as an operator
\[
\tp: \CH_\alpha^{\otimes n_2} \otimes \CH_\beta^{\otimes n_1} \to \CH_\sigma
\]
that implements the following mapping
\begin{align}
    \alpha_{ij}^x \beta_{ij}^x &\;\mapsto\; \sigma_{ij}^x ,\\
    \prod_{i\in \delta^1(a)} \alpha_{ij}^z &\;\mapsto\; \prod_{i\in \delta^1(a)}\sigma_{ij}^z ,\\
    \prod_{j\in \delta^2(b)} \beta_{ij}^z &\;\mapsto\; \prod_{j\in \delta^2(b)}\sigma_{ij}^z ,
\end{align}
where $\alpha_{ij}$ denote qubits in the $j$th copy of code $\CC_1$ (stacked $n_2$ times), and $\beta_{ij}$ denote qubits in the $i$th copy of code $\CC_2$ (stacked $n_1$ times).
The ZX diagram and quantum process realization are shown in \figref{fig:tpzx}.
\begin{figure}[ht]
    \centering
    \includegraphics[width=\linewidth]{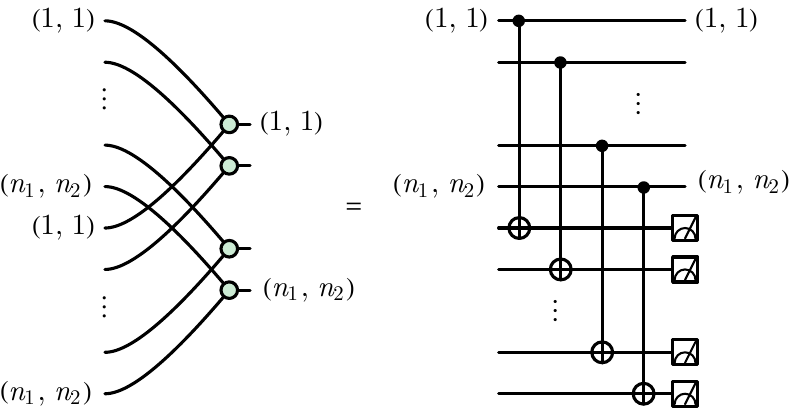}
    \caption{ZX diagram and quantum process realization of the tensor product construction.}
    \label{fig:tpzx}
\end{figure}

\subsubsection{Phase of matter}
The tensor product construction admits a natural interpretation as an operation on operator algebras, realized physically as a coupled-layer model. Starting from the operator algebras $\CA_{\CC_1}$ and $\CA_{\CC_2}$, we form $n_2$ copies of $\CA_{\CC_1}$ and $n_1$ copies of $\CA_{\CC_2}$, and define the stacked operator algebra $\CA_{\mathrm{stack}}$ as the algebra generated by the union of their generators, acting on the joint Hilbert space $\CH_\alpha^{\otimes n_2} \otimes \CH_\beta^{\otimes n_1}$. 

The canonical Hamiltonians associated with these stacked algebras are precisely the layer Hamiltonians: we stack $n_2$ copies of $H_1$, labeling the $j$th copy as
\begin{align}
 H_1^{(j)}=- \sum_{a \in A}C^1_{a,j} - h_1\sum_i \alpha_{ij}^x.   
\end{align}
Similarly, we stack $n_1$ copies of $H_2$, labeling the 
$i$th copy as 
\begin{align}
H_2^{(i)}=- \sum_{b \in B}C^2_{i,b} - h_2\sum_j \beta_{ij}^x. 
\end{align}
Then, we introduce a coupling Hamiltonian
\begin{align}
    H_{\mathrm{cpl}}=-\lambda\sum_{i,j}\alpha_{ij}^z\beta_{ij}^z,
\end{align}
and couple together the stacks to get
\begin{align}
    H=\sum_j H_1^{(j)}+\sum_{i}H_2^{(i)}+H_\mathrm{cpl}.
\end{align}
In the strong-coupling limit $\lambda \gg 1,h_1,h_2$, the low-energy subspace is characterized by
the local constraint $\alpha_{ij}^z = \beta_{ij}^z$.
We define: 
\begin{align}
    \sigma_{ij}^z &\coloneqq \alpha_{ij}^z=\beta_{ij}^z,\\
    \sigma_{ij}^x &\coloneqq \alpha_{ij}^x\beta_{ij}^x.
\end{align}
The quantum process $\tp$ thus acts as a \emph{quotient map on operator algebras}: it identifies $\alpha_{ij}^z = \beta_{ij}^z \eqqcolon \sigma_{ij}^z$ and maps $\alpha_{ij}^x \beta_{ij}^x \mapsto \sigma_{ij}^x$, projecting $\CA_{\mathrm{stack}}$ onto the operator algebra $\CA_{\CC_1 \otimes \CC_2}$ of the tensor product code. The projection is precisely the strong-coupling limit of $H_{\mathrm{cpl}}$, and the resulting operator algebra generates the tensor product code Hamiltonian. 

At second order in perturbation theory, the effective Hamiltonian reduces to
\begin{equation}\label{eq:tpeff}
\begin{split}
    H_{\mathrm{eff}} \simeq - \sum_{a \in A, j}C^1_{a,j} - \sum_{i, b \in B}C^2_{i,b} - \frac{h_1 h_2}{\lambda} \sum_{i,j} \sigma_{ij}^x.
\end{split}
\end{equation}
This effective Hamiltonian reproduces exactly the row and column checks of the tensor product code $\mathcal{C}_1 \otimes \mathcal{C}_2$ in Eq.~(\ref{eq:tensorchecks}), confirming that the coupled-layer construction yields the correct operator algebra $\CA_{\CC_1 \otimes \CC_2}$. The transverse-field term emerges at second order with effective coupling $h_{\text{eff}}=h_1h_2/\lambda$. If $\mathcal{C}_1$ and $\mathcal{C}_2$ correspond to symmetry-broken phases, the coupled system remains symmetry-broken but acquires local redundancies $\prod_{j \in \delta^2(b)} C^1_{a,j}\prod_{i \in \delta^1(a)} C^2_{i,b}$ for all pairs of $(a,b)$. Thus, subsequent application of the KW gauging procedure yields topological order. Detailed perturbative derivations and explicit examples are given in \appref{app:perttensor}. 
\subsection{Check product code}
In contrast to the tensor product, where row and column checks act independently, the check product couples the two codes so that each check involves bits from both. The check product code $\mathcal{C}_1\star \mathcal{C}_2$ has parity check matrix 
\begin{align}
    \BH_{\mathcal{C}_1\star \mathcal{C}_2} = \begin{bmatrix}
        \BH_{\mathcal{C}_1} \otimes \BH_{\mathcal{C}_2}
    \end{bmatrix}_{m_1m_2\times n_1n_2}.
\end{align}
Let $\CH_\gamma$ be the embedding Hilbert space for $\CC_1 \star \CC_2$. The check operators of the check product code are given as
\begin{align} \label{eq:checkchecks}
    C_{a,b}=\prod_{i \in \delta^1(a), j \in \delta^2(b)} \gamma^z_{ij}.
\end{align}
for each pair $(a,b) \in A \times B$.
The check product code has two types of logical operators: row logicals $L_{\lambda,j}^1$ inherited from $\CC_1$ acting on each row $j$, and column logicals $L_{i,\nu}^2$ inherited from $\CC_2$ acting on each column $i$, where 
\begin{align}
    L_{\lambda,j}^1 = \prod_{i\in \lambda} \gamma^x_{ij}, \quad L_{i,\nu}^2 = \prod_{j\in \nu} \gamma^x_{ij}.
\end{align}
% i.e. we place $L_p^{(1)}$ logical operator at $i^{th}$ row, and similarly for $L^{(2)}_{q,j}$ we place $L^{(2)}_q$ logical operator in the $j^{th}$ column.
\subsubsection{ZX realization as quantum process}
The check product can be realized as an operator
\[
\chp : \CH_\alpha^{\otimes n_2} \otimes \CH_\beta^{\otimes n_1} \to \CH_\gamma
\]
that implements the following mapping
\begin{align}
    \alpha_{ij}^x &\mapsto \gamma_{ij}^x ,\\
    \beta_{ij}^x  &\mapsto \gamma_{ij}^x ,\\
    \alpha_{ij}^z \beta_{ij}^z &\mapsto \gamma_{ij}^z .
\end{align}
The ZX diagram and quantum process realization are shown in \figref{fig:cpzx}.

\begin{figure}[ht]
    \centering
    \includegraphics[width=\linewidth]{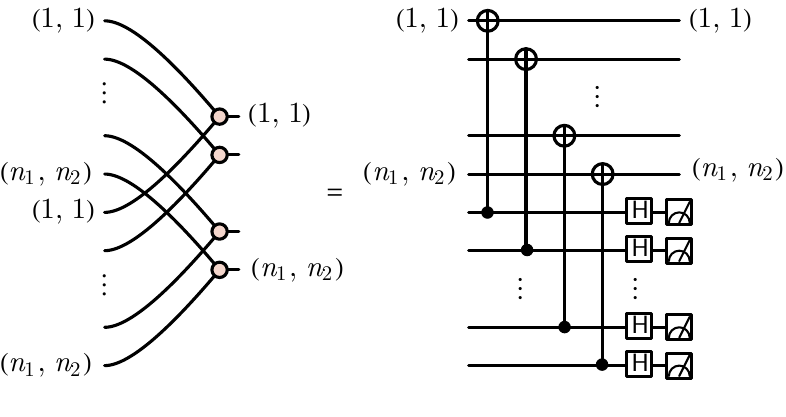}
    \caption{ZX diagram and quantum process realization of the check product construction.}
    \label{fig:cpzx}
\end{figure}

Note that the ZX diagram for check product is a dual-colored version of the tensor-product diagram. Thus, it is evident from the ZX diagram that $\chp = \sG\, \tp\, \sG$, where $\sG$ denotes the layer of Hadamard gates on the input and output qubits respectively, which is consistent with the identity
\begin{align*}
    \CC_1 \star \CC_2 = (\CC_1^\perp \otimes \CC_2^\perp)^\perp .
\end{align*}
\subsubsection{Phase of matter}
The check product admits an analogous operator-algebra interpretation, dual to the tensor product case. We again form the stacked operator algebra $\CA_{\mathrm{stack}}$, but now introduce an $X$-type interlayer coupling
\begin{align}
    \tilde{H}_{\mathrm{cpl}}=-\lambda \sum_{i,j}\alpha_{ij}^x \beta_{ij}^x .
\end{align}
Coupling together the stacks yields the total Hamiltonian
\begin{align}
    \tilde{H}=\sum_j H_1^{(j)}+\sum_{i}H_2^{(i)}+\tilde{H}_\mathrm{cpl}.
\end{align}
In the strong-coupling limit $\lambda \gg 1,h_1,h_2$, the low-energy subspace is characterized
by the constraint $\alpha_{ij}^x = \beta_{ij}^x$.
Then we can define
\begin{align}
    \gamma_{ij}^x &\coloneqq \alpha_{ij}^x=\beta_{ij}^x,\\
    \gamma_{ij}^z &\coloneqq \alpha_{ij}^z\beta_{ij}^z.
\end{align}
The quantum process $\chp$ again acts as a quotient map on operator algebras, but with the roles of $X$ and $Z$ exchanged relative to the tensor product: it identifies $\alpha_{ij}^x = \beta_{ij}^x \eqqcolon \gamma_{ij}^x$ and maps $\alpha_{ij}^z \beta_{ij}^z \mapsto \gamma_{ij}^z$, projecting $\CA_{\mathrm{stack}}$ onto $\CA_{\CC_1 \star \CC_2}$. This duality between the two product constructions is manifest in the ZX-diagram representation, where the check product diagram is obtained from the tensor product diagram by exchanging Z- and X-spiders.

Unlike the tensor product case, the leading nontrivial contributions to the effective Hamiltonian arise only at higher order in perturbation theory.
At order $|\delta^1(a)| + |\delta^2(b)|$, the effective Hamiltonian takes the form
\begin{equation}\label{eq:cpeff}
\begin{split}
    \tilde{H}_{\mathrm{eff}}=&-(h_1+h_2)\sum_{ij}\gamma_{ij}^x \\& + \frac{\kappa}{(-\lambda)^{|\delta^1(a)|+|\delta^2(b)|}} \sum_{a \in A, b \in B} C_{a,b},
\end{split}
\end{equation}
where $\kappa$ is a constant.

This effective Hamiltonian reproduces exactly the checks of the check product code $\mathcal{C}_1 \star \mathcal{C}_2$ in Eq.~(\ref{eq:checkchecks}), confirming that the coupled-layer construction yields the correct operator algebra $\CA_{\CC_1 \star \CC_2}$. The transverse-field terms contribute at first order with the effective coupling $h_{\text{eff}}=h_1 + h_2$, while the checks arise only at higher order — in contrast to the tensor product case where the roles are reversed. Physically, the check product realizes a subsystem symmetry or fracton-like phase, consistent with earlier coupled-layer constructions of fracton order \cite{ma_fracton_2017, vijayIsotropicLayerConstruction2017,tanFractonModelsProduct2025}. Explicit perturbative derivations and concrete examples are provided in \appref{app:pertcheck}.

\section{Generalization to PQ-product \label{sec:pq-prod}}
We generalize the tensor product and check product to a general product which we call the $(p,q)$-product. Familiar cases include the tensor product $(p,q)=(2,1)$, the check product $(2,2)$, and the cubic product $(3,2)$; the general definition is given below. 
Let $\{\CC_s\}_{s=1}^p$ be $p$ classical codes. Let $[p]=\{1, ..., p\}$. The $(p,q)$-product code can be constructed in two steps:
\begin{enumerate}
    \item For every $q$-element subset $S \subseteq [p]$, construct the check product \[
    \CC_S=\bigstar_{s \in S} \CC_s .
    \]
    \item Take the tensor product over all such subsets:\[
    \CC_{(p,q)}=\bigotimes_{S \subseteq [p], |S|=q}\CC_S .
    \]
\end{enumerate}

To define the $(p,q)$-product more concretely in terms of the parity matrix, let $\CC_s$ have parity check matrices $\BH_{\CC_s} \in \BF_2^{m_s \times n_s}$
, boundary maps $\delta^{(s)}$, and sets of checks $A_s=\{C_{a_s}^{(s)}\}$, embedding Hilbert spaces $\CH_{\alpha^{(s)}}$, and associated Hamiltonians 
\begin{align}
    H_{\alpha^{(s)}} = -\sum_{a_s \in A_s} C_{a_s}^{(s)}-h_s\sum_{i_s}\alpha_{i_s}^{(s)x},
\end{align}
where $C_{a_s}^{(s)}=\prod_{i_s \in \delta^{(s)}(a_s)}\alpha_{i_s}^{(s)z}$. As before, these are the canonical Hamiltonians associated with the operator algebras $\CA_{\CC_s}$ via Eq.~(\ref{eq:cH}).

Let $r = \binom{p}{q}$. Denote all $q$-element subsets of $[p]$ by 
\[
\binom{[p]}{q} \;=\; \{S\subseteq [p] : |S|=q\}.
\]
Define, for each $S\in \binom{[p]}{q}$, a block matrix
\begin{equation*}
\BH_S \;:=\; \bigotimes_{s=1}^p M_s(S),
\qquad
M_s(S) \;=\;
\begin{cases}
\BH_{\CC_s},& s\in S.\\
\mathbb{I}_{n_s},& s\notin S.
\end{cases}
\end{equation*}
Then, the ${(p,q)}$-product has the parity check matrix formed by vertical concatenation of these blocks:
\begin{equation*}
\BH_{(p,q)} \;:=\; \begin{bmatrix} \BH_{S_1} \\ \BH_{S_2} \\ \vdots \\ \BH_{S_r} \end{bmatrix},
\end{equation*}
where $\{S_1,\dots,S_r\}=\binom{[p]}{q}$ in any fixed order.

This is a generalization of familiar constructions:
\begin{itemize}
    \item $(p,q)=(2,1)$ gives the tensor product code
    \[
    \BH_{(2,1)} = \begin{bmatrix}
        \mathbb{I}_{n_2}\otimes \BH_{\mathcal{C}_1} \\
        \BH_{\mathcal{C}_2}\otimes \mathbb{I}_{n_1}
    \end{bmatrix}
    \]
    \item $(p,q)=(2,2)$ gives the check product code
    \[
    \BH_{(2,2)} = \begin{bmatrix}
        \BH_{\mathcal{C}_1} \otimes \BH_{\mathcal{C}_2}
    \end{bmatrix}
    \]
    \item $(p,q)=(3,2)$ gives the cubic product code
    \[\BH_{(3,2)} =
    \begin{bmatrix}
        \mathbb{I}_{n_1} \otimes \BH_{\mathcal{C}_2} \otimes \BH_{\mathcal{C}_3} \\
        \BH_{\mathcal{C}_1} \otimes \mathbb{I}_{n_2}\otimes \BH_{\mathcal{C}_3} \\
        \BH_{\mathcal{C}_1}\otimes \BH_{\mathcal{C}_2}\otimes \mathbb{I}_{n_3}
    \end{bmatrix}
    \]
\end{itemize}
The ZX-diagram realization and Hamiltonian interpretation of the $(p,q)$-product follow directly from its definition as a composition of check and tensor products. The coupled-layer interpretation and operator-algebra quotient structure of the tensor and check products extend naturally to the general $(p,q)$-product. We illustrate this explicitly for the cubic $(p,q)=(3,2)$ case in Appendix~\ref{app:pq}.

\section{Discussion and outlook \label{sec:discussion}} 

In this paper, we studied several fundamental LDPC code transformations, including generalized Kramers-Wannier duality, tensor product, check product, and their generalization to the $(p,q)$-product through a unified operational framework. Central to our approach is the observation that the physical content of a cLDPC code is captured by the operator algebra associated with its Tanner graph, and that code transformations are most naturally understood as maps between operator algebras, from which Hamiltonian mappings and phase structures follow as consequences. For each transformation, we provided three complementary realizations: as a ZX diagram, as an explicit quantum circuit consisting of ancilla initialization, local unitaries, and projective measurements, and as a mapping between quantum phases of matter.

For KW duality, the ZX diagram arises directly from the Tanner graph of the code. Using this structure, we provided a systematic algorithm to extract quantum circuits and identified fundamental lower bounds on ancilla and measurement overhead in terms of the code's symmetries and redundancies. Different circuit realizations correspond naturally to distinct physical pictures of gauging — minimal coupling and defect condensation — which clarify how KW duality maps trivial phases to SSB or topologically ordered phases depending on the presence of local redundancies. These results provide explicit and scalable protocols for preparing states in nontrivial quantum phases.

For tensor and check product constructions, the ZX diagrams and circuits naturally realize coupled-layer constructions familiar from many-body physics. At the operator-algebra level, these constructions act as quotient maps on stacked operator algebras, projecting onto the product code's algebra via interlayer coupling. We made this correspondence precise by deriving the effective product-code Hamiltonians using perturbation theory in the strong-coupling limit, which serves as a consistency check that the correct operator algebra is reproduced.

More broadly, viewing code transformations as physical quantum processes provides a common language for understanding logical operations as phase-changing mechanisms. From this perspective, gauging and coupled-layer constructions emerge as elementary building blocks that interpolate between distinct quantum phases in a controlled and operationally meaningful way. This connection highlights a deep relationship between quantum error correction and quantum phase transitions \cite{dennisTopologicalQuantumMemory2002}, and suggests new avenues for state preparation \cite{ tantivasadakarnHierarchyTopologicalOrder2023,tantivasadakarnShortestRouteNonAbelian2023,bravyiLiebRobinsonBoundsGeneration2006}, logical encoding \cite{williamson_low-overhead_2024}, and resource generation.

This framework opens several directions for future investigation. A natural next step is to extend the operator-algebra framework to qLDPC codes, where the interplay between X- and Z-type checks gives rise to a richer algebraic structure. Such an extension could provide new insights into the phase structure and logical operations of qLDPC codes. Our framework can also be extended to other code constructions, such as balanced product codes \cite{breuckmann_balanced_2021} and related homological products \cite{bravyiHomologicalProductCodes2013,tillich_quantum_2014,hastingsFiberBundleCodes2021}, potentially yielding new families of quantum phases. On the practical side, an important question concerns the robustness of the proposed circuits under experimental imperfections in ancilla initialization, unitaries, and measurements, and how error correction can be incorporated into these processes. Finally, the flexibility of the ZX-calculus suggests systematic optimization strategies — trading off locality, ancilla overhead, and measurement cost — by exploring families of equivalent ZX diagrams, an approach that may be particularly valuable for near-term implementations \cite{kissingerReducingNumberNonClifford2020}.

{\it Note.}
Upon finishing our work we became aware of a related paper \cite{zhangCoupledLayerConstructionQuantum2026} on coupled layer construction of product codes.

\begin{acknowledgments}
We acknowledge inspiring discussions with Tibor Rakovszky and Vedika Khemani. This research is supported by NSF Grant No.~DMR-2238360.
\end{acknowledgments}

\bibliographystyle{apsrev4-2}
\bibliography{references.bib}

\onecolumngrid
\newpage
\appendix
\section{Quantum process realization of \texorpdfstring{$\sD$}{D} for 1D Ising model} \label{app:I1D}
For simplicity, we ignore overall phases and scalar coefficients in the ZX-calculus derivation below. These factors do not affect the structure of the quantum process or the mapping between states, and can be restored systematically if needed.
\subsection{Quantum Process Extraction from ZX-Diagram}\label{app:I1D-extract}
Consider a cLDPC code $\CI_{1D}$ based on the Ising model defined on a 1D lattice of three sites. The Tanner graph and ZX-diagram representation of $\sD$ are shown below:
\begin{figure}[ht]
    \centering
    \includegraphics[width=0.6\linewidth]{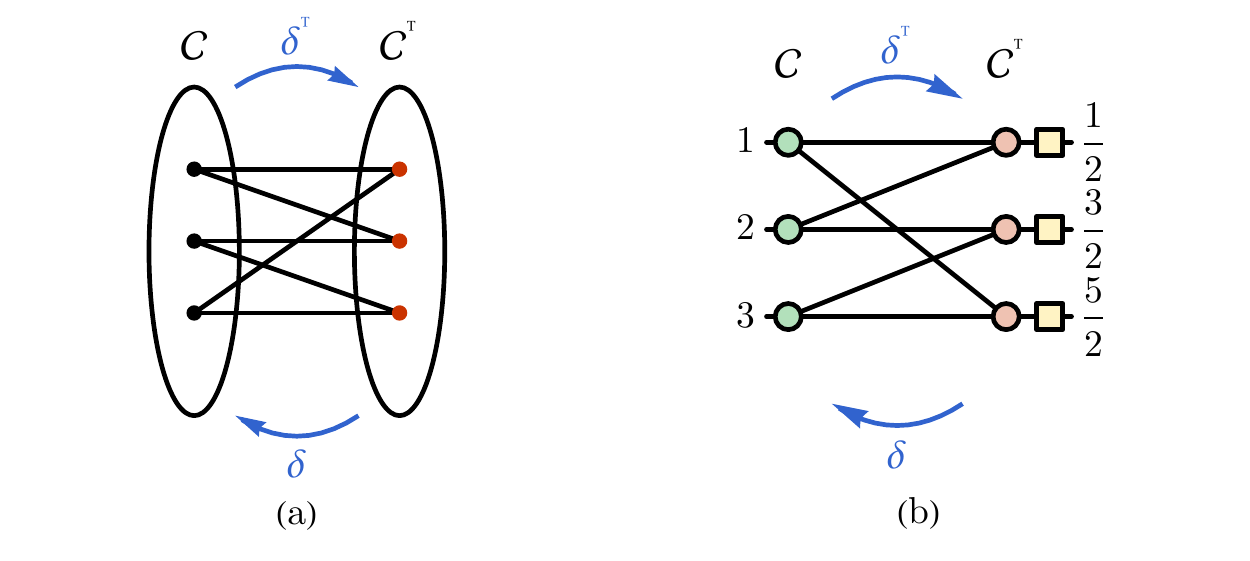}
    \caption{(a) Tanner graph representation of $\CI_{1D}$.
(b) ZX-diagram representation of $\sD$.}
    \label{fig:I1DTannerZX}
\end{figure}

The generalized KW transformation in \eqnref{eq:D1} and \eqnref{eq:D2} can be explicitly demonstrated on ZX-diagram using ZX-calculus rules: 
\begin{figure}[ht]
    \centering
    \includegraphics[width=0.95\linewidth]{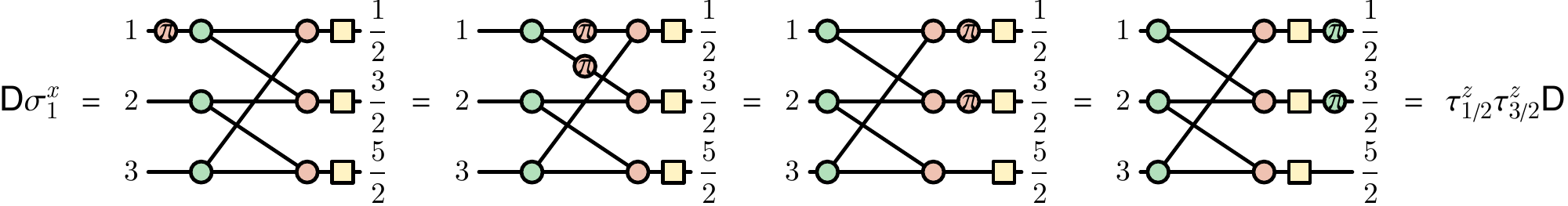}\\
    \includegraphics[width=0.95\linewidth]{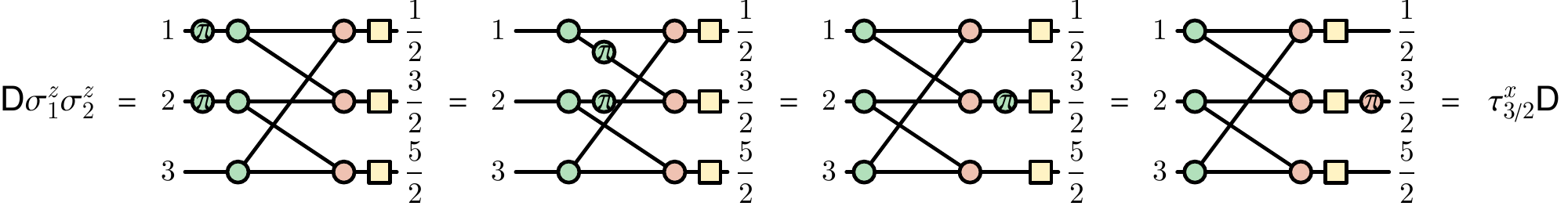}
    \label{fig:I1DKW}
\end{figure}

To extract a quantum circuit from the ZX-diagram of $\sD$ in \figref{fig:I1DTannerZX}(b), we follow the steps of Algorithm \ref{alg:process-extraction}:
\begin{enumerate}
    \item Perform Gaussian elimination on $\BH=\begin{pmatrix}
    1 & 0 & 1 \\
    1 & 1 & 0 \\
    0 & 1 & 1 
\end{pmatrix}$.
    \begin{figure}[ht]
        \centering
        \includegraphics[width=0.6\linewidth]{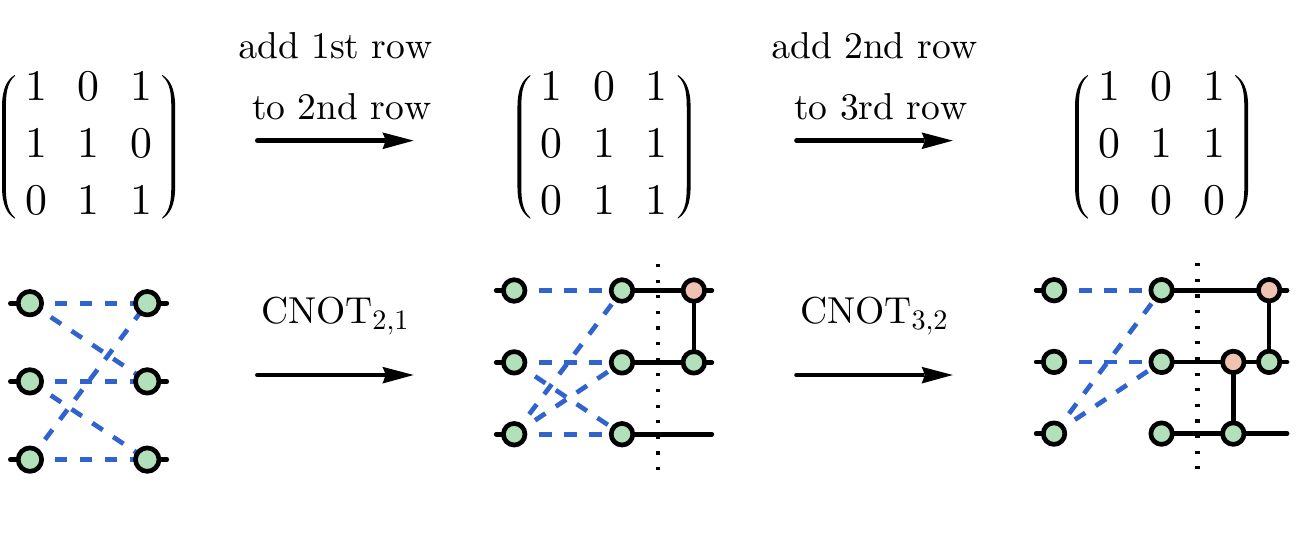}
    \end{figure}
    \item Perform Gaussian elimination on $\begin{pmatrix}
    1 & 0 & 1 \\
    0 & 1 & 1 \\
    0 & 0 & 0 \end{pmatrix}^\top=\begin{pmatrix}
    1 & 0 & 0 \\
    0 & 1 & 0 \\
    1 & 1 & 0 \end{pmatrix}$.
    \begin{figure}
        \centering
        \includegraphics[width=0.7\linewidth]{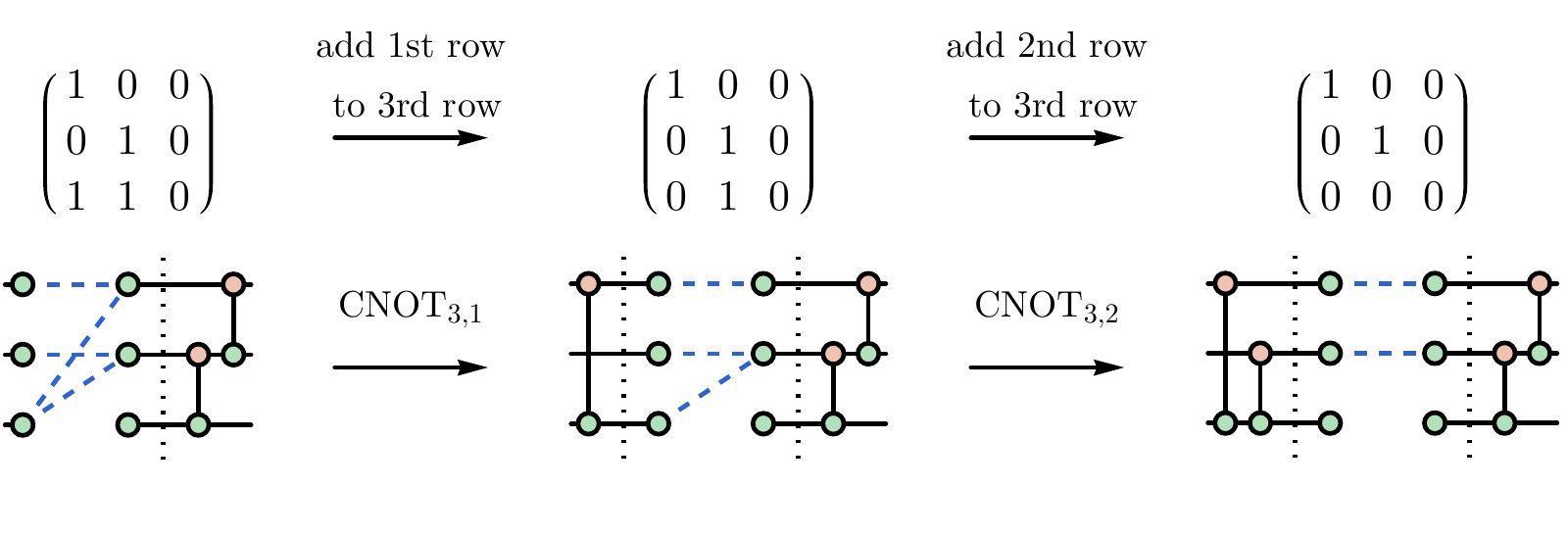}
    \end{figure}
\end{enumerate}

The extracted circuit is shown below. Besides the \texttt{CNOT} and Hadamard gates, it involves 1 ancilla and 1 measurement, which is the number of redundancies and symmetries for $\CI_{1D}$, respectively.  
\begin{figure}[ht]
    \centering
    \includegraphics[width=0.6\linewidth]{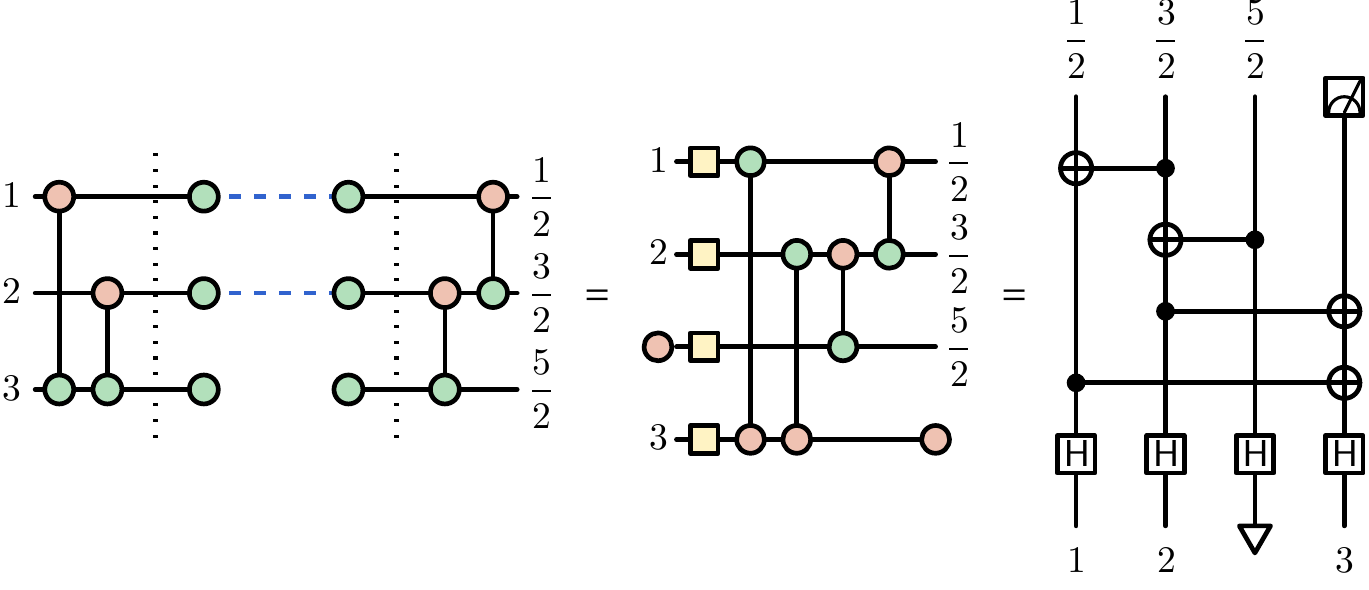}
    \caption{Circuit realizing $\sD$ for the 1D Ising example, extracted from the ZX diagram via Algorithm~\ref{alg:process-extraction}.}
    \label{fig:I1D-extracted-circuit}
\end{figure}

\subsection{Minimal Coupling vs Defect Condensation} \label{app:I1D-mc-vs-dc}
The above realization of $\sD$ corresponds to the gauging via defect condensation picture. On the other hand, the gauging via minimal coupling picture gives a circuit involving 3 ancilla and 3 measurements.
\begin{figure}[h]
    \centering
    \includegraphics[width=0.47\linewidth]{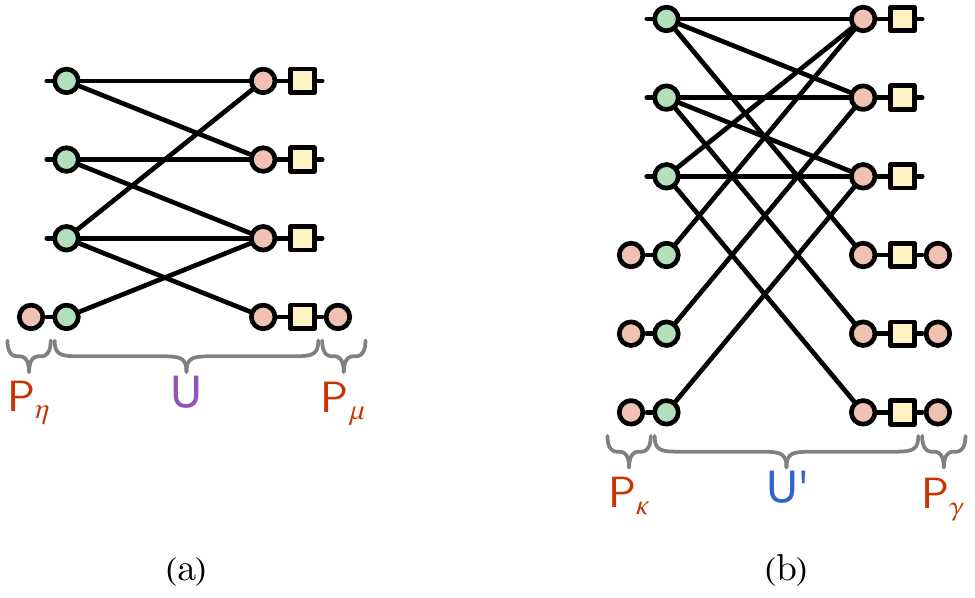}
    \caption{(a) Defect condensation: $k^\top$ ancilla qubits and $k$ measurements.
(b) Minimal coupling: $m$ ancilla qubits and $n$ measurements.}
    \label{fig:I1D-dc-vs-mc}
\end{figure}

The extracted minimal coupling circuit is shown below. 
\begin{figure}[h]
    \centering
    \includegraphics[width=0.2\linewidth]{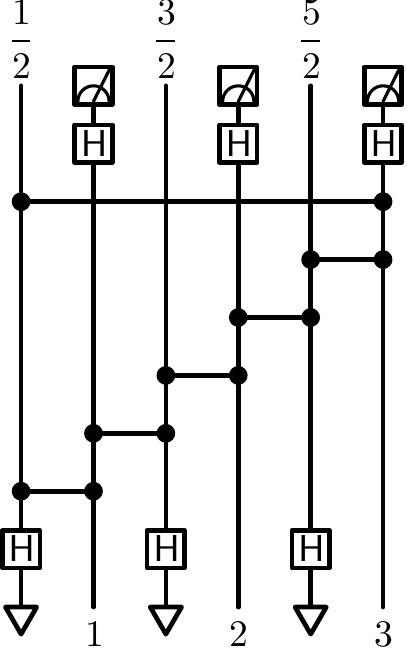}
    %\caption{Caption}
    %\label{fig:enter-label}
\end{figure}
\subsection{State Preparation in SSB Phase} \label{app:stateprep}

In \secref{sec:stateprep}, we showed that the circuit extracted for $\sD$
can be used to prepare SSB states starting from a trivial product state. By adjusting the basis and relative phase of the ancilla qubits, the
construction provides explicit control over which ground state within the SSB manifold is prepared. Here, we illustrate this mechanism using the defect
condensation realization of $\sD$ for the three-qubit $\CI_{1D}$.

Starting from the product state $\ket{+++}$, different choices of the ancilla
state $\ket{s}$ allow us to prepare any ground state in the SSB manifold spanned
by $\ket{000}$ and $\ket{111}$. We denote by $\sD(\ket{s})$ the defect
condensation circuit realization of $\sD$ with ancilla initialized in the state
$\ket{s}$. It suffices to verify that
\[
\sD(\ket{+})\ket{+++} = \ket{000}, \qquad
\sD(\ket{-})\ket{+++} = \ket{111},
\]
as we verify using ZX-calculus below.
\begin{figure}[h]
    \centering
    \includegraphics[width=0.48\linewidth]{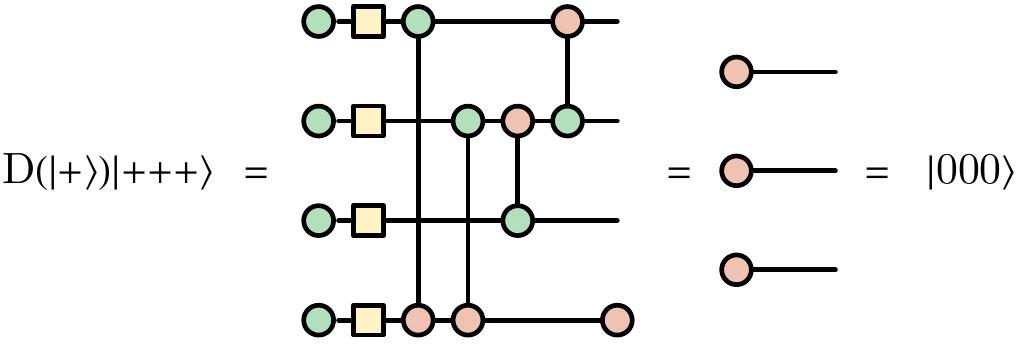} \quad
    \includegraphics[width=0.48\linewidth]{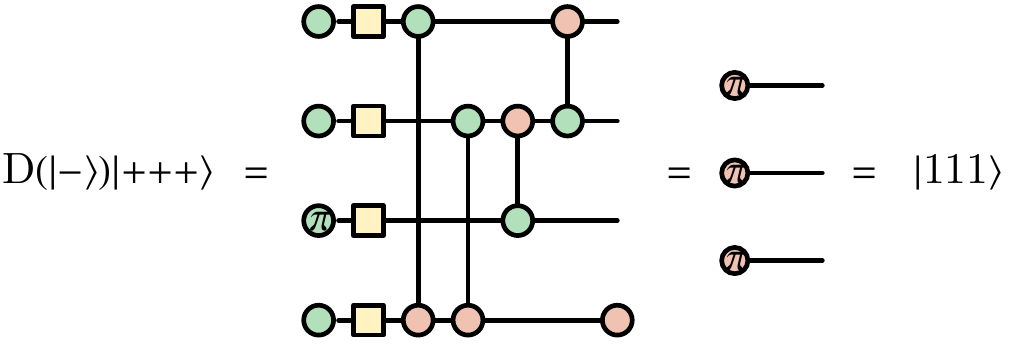}
    \label{fig:stateprep}
\end{figure}
By linearity, choosing
$\ket{s} = \alpha \ket{+} + \beta \ket{-}$ prepares the state
$\alpha\ket{000} + \beta\ket{111}$ for arbitrary coefficients
$\alpha,\beta$ satisfying $|\alpha|^2 + |\beta|^2 = 1$.

\section{Gauging via minimal coupling vs Gauging via defect condensation}\label{app:gauge}

Consider a Hamiltonian defined by a cLDPC code $\CC$ with $n$ bits and $m$ checks:
\begin{align}
    H=-J\sum_{a=1}^m \prod_{i \in \delta(a)} \sigma_i^z-h\sum_{i=1}^n\sigma_i^x
    \label{eq:sigmaH}
\end{align}

In the gauging via minimal coupling picture \cite{barkeshli_symmetry_2019,harlow_symmetries_2019}, one gauge field $\kappa_a^z$ is coupled to each check: 
\begin{align}
    H_{\mathrm{gauge}}=-J\sum_{a=1}^m \kappa_a^z\prod_{i \in \delta(a)} \sigma_i^z-h\sum_{i=1}^n\sigma_i^x
\end{align}
The introduction of $m$ gauge fields $\kappa_a^z$ is equivalent to attaching $m$ ancilla qubits.

Next we impose $n$ Gauss laws
\begin{align}
    \gamma_i=\sigma_i^x\prod_{a\in \delta^\top(i)} \kappa_a^x=1,
\end{align}
which imposes a restriction of $\CH_{\sigma,\kappa}$ to a gauge invariant physical subspace $\CH_{\sigma,\kappa}^{\mathrm{sym}}$. The enforcement of $n$ Gauss laws is equivalent to conducting $n$ projective measurements. 

We can rewrite $H_{\mathrm{gauge}}$ in terms of new variables that are manifestly gauge-invariant:
\begin{align}
    H_{\mathrm{gauge}}=-J\sum_a \tau_a^x-h\sum_i \prod_{a \in \delta^\top(i)} \tau_a^z \label{eq:tauH}
\end{align}
where
\begin{align}
    \tau_a^x = \kappa_a^z \prod_{i \in \delta(a)} \sigma_i^z, \quad \tau_a^z = \kappa_a^x.
\end{align}
This redefinition of variables can be achieved exactly by a unitary circuit $\mathsf{U}'$.

In this way, gauging via minimal coupling implements the KW duality transformation, mapping \eqnref{eq:sigmaH} to \eqnref{eq:tauH}. This transformation can be captured explicitly by the operator $\sD=\sP_{\gamma} \sU'\sP_\kappa$, where 
\begin{itemize}
    \item $\sP_\kappa$ represents the attachment of $m$ ancilla qubits required for minimal coupling,
    \item $\sU'$ is the unitary circuit that achieves the transformation to gauge-invariant variables, and 
    \item $\sP_\gamma$ enforces $n$ Gauss law constraints through projective measurements.
\end{itemize}
Thus, each step in the minimal coupling gauging procedure has a direct operational interpretation within the resource-intensive quantum process that realizes $\sD$. This makes the connection between the duality map and the gauging procedure fully explicit.

\begin{figure}[ht]
    \centering
    \includegraphics[width=0.5\linewidth]{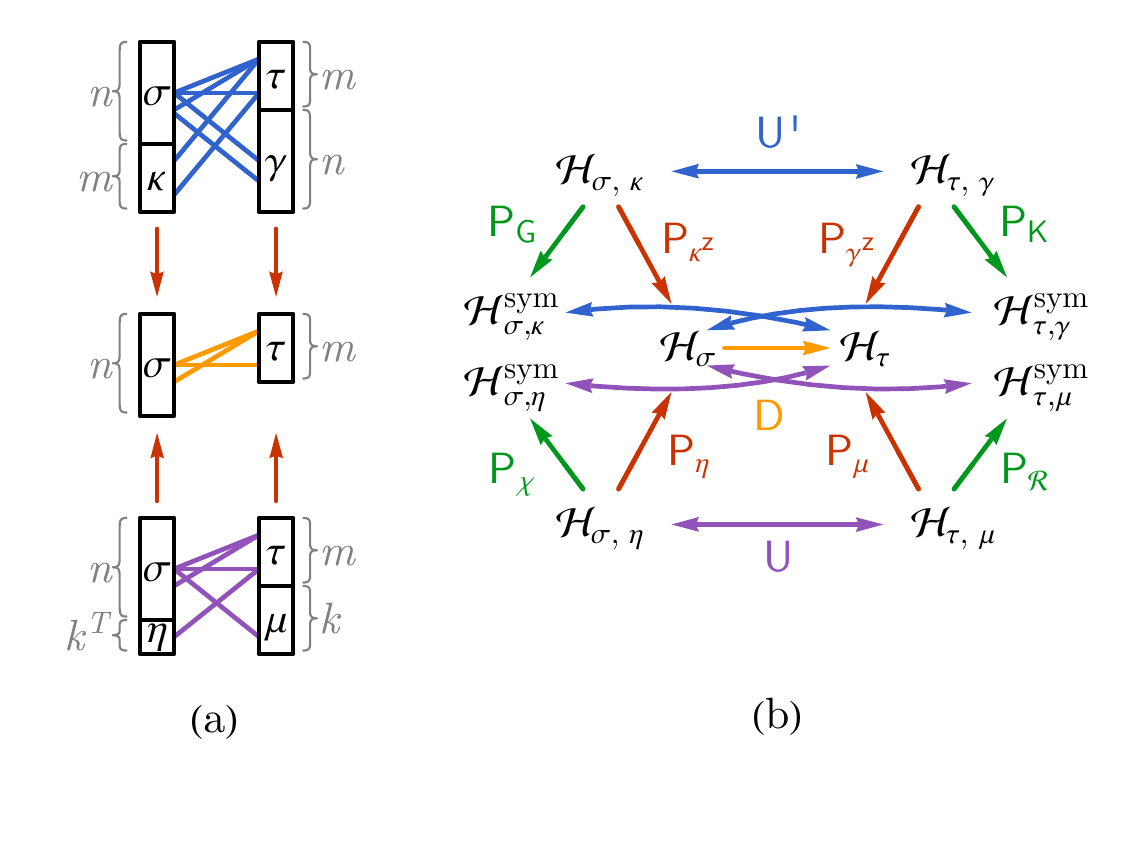}
    \caption{Two quantum process realizations of the KW duality map $\sD$. 
(a)~ZX-diagram representations. Middle: the non-invertible map $\sD : \CH_\sigma \to \CH_\tau$. Top: minimal coupling realization, embedding $\sD$ into a unitary $\sU'$ acting on an enlarged Hilbert space with $m$ ancilla qubits ($\kappa$) and $n$ projective measurements ($\gamma$). Bottom: defect condensation realization, embedding $\sD$ into a unitary $\sU$ with $k^\top$ ancilla qubits ($\eta$) and $k$ projective measurements ($\mu$). 
(b)~Commutative diagram of Hilbert spaces and maps. The non-invertible map $\sD$ (center, yellow) factors through either the minimal coupling unitary $\sU'$ (top, blue) or the defect condensation unitary $\sU$ (bottom, purple), related by ancilla projections and measurement projections on each side.}
    \label{fig:mc-vs-dc}
\end{figure}

Another way to perform gauging is through defect condensation \cite{roumpedakis_higher_2023,gorantla_tensor_2024}, where we insert and sum over topologically distinct symmetry defects — a procedure often referred to as orbifolding in 2D. 

Starting with the Hamiltonian in \eqnref{eq:sigmaH}, we introduce $k^\top$ defect variables $\eta_\beta$, one for each redundancy $R_\beta \in \ker(\delta)$, which correspond to topologically distinct symmetry defects. To couple the defects to the Hamiltonian, we first fix a basis $\{R_\beta\}_{\beta=1}^{k^\top}$ of the redundancy space $\ker(\delta)$, and construct a dual basis $\{G_\beta\}_{\beta=1}^{k^\top}$ of check subsets satisfying \cite{rakovszky_physics_2024}
\begin{align}
    \langle G_\beta | R_{\beta'} \rangle = \delta_{\beta\beta'},
\end{align}
where $\langle G | R \rangle = \sum_{a \in G} |R \cap \delta(a)| \mod 2$ counts the parity of overlap. Each $G_\beta$ is a set of checks that overlaps with $R_\beta$ an odd number of times and with all other redundancies an even number of times. In many cases of interest, the dual basis can be chosen such that the sets $\{G_\beta\}$ are mutually disjoint, i.e., each check $a$ belongs to at most one $G_\beta$. With this choice, the defect-coupled Hamiltonian is obtained by modifying the check operators $C_a \to \hat{C}_a$, where
\begin{align}
    \hat{C}_a = 
    \begin{cases}
        \eta_\beta \, C_a, & \text{if } a \in G_\beta, \\
        C_a, & \text{if } a \notin \bigcup_\beta G_\beta.
    \end{cases}
\end{align}
Each defect variable $\eta_\beta$ couples to exactly the checks in $G_\beta$, and each check is modified by at most one defect variable. The resulting Hamiltonian is
\begin{align}
    H_{\mathrm{gauge}} = -J \sum_{a} \hat{C}_a - h \sum_i \sigma_i^x.
\end{align}
By construction, the modified checks satisfy $\prod_{a \in R_\beta} \hat{C}_a = \eta_\beta$ for each redundancy, so that the defect variables $\eta_\beta$ label distinct symmetry sectors. The insertion of $k^\top$ defect variables corresponds to attaching $k^\top$ ancilla qubits.

Then, we project onto gauge invariant subspace $\CH_{\sigma,\eta}^{\mathrm{sym}}$ by enforcing that all $k$ logical operators defined by the symmetries $\{\lambda_\alpha\}_{\alpha=1}^{k}$ of $\CC$ have eigenvalue $+1$. For each $\alpha=1,\dots,k$,

\begin{align}
   \mu_{\lambda_\alpha}=\prod_{i \in \lambda_\alpha}\sigma_i^x=1.
\end{align}
This enforcement of $k$ constraints corresponds to performing $k$ projective measurements. 

We can rewrite $H_{\mathrm{gauge}}$ in the form of \eqnref{eq:tauH} through a transformation to new gauge-invariant variables. Using the modified checks $\hat{C}_a$ defined above, we identify
\begin{align}
    \tau_a^x &=
    \hat{C}_a = 
    \begin{cases}
        \eta_\beta \, C_a, & \text{if } a \in G_\beta, \\
        C_a, & \text{if } a \notin \bigcup_\beta G_\beta.
    \end{cases}, \quad \tau_a^z = \left(\prod_{\beta:\, a \in R_\beta} \eta_\beta\right) \prod_{i \in \Delta_a} \sigma_i^x,
\end{align}
where $\Delta_a$ is determined by the condition that $\tau_a^z$ commutes with all $\hat{C}_{a'}$ and anticommutes only with $\hat{C}_a$. This condition can be solved explicitly given the dual basis $\{G_\beta\}$ and the code structure. The redefinition of variables is achieved by the unitary circuit $\sU$.

Therefore, the gauging process via defect condensation is captured by the quantum process $\sD = \sP_\mu \, \sU \, \sP_\eta$, involving the attachment of $k^\top$ ancilla qubits, a unitary circuit $\sU$, and $k$ projective measurements.
% We can rewrite $H_{\mathrm{gauge}}$ in the form of \eqnref{eq:tauH} through a transformation to new gauge-invariant variables:
% \begin{align}
%     \tau_a^x &=\prod_{\alpha \in \delta(a)}\eta_\alpha \prod_{i \in \delta(a)}\sigma_i^z\\
%     \tau_a^z &=\prod_{\alpha \in \CR_a} \eta_\alpha \prod_{i \in A : \delta^\top(A)=\CB_a} \sigma_i^x
% \end{align}
% where $\CR_a=\{\alpha \,|\, a \in R_\alpha\}$, $\CB_a=\bigominus_{\alpha \in \CR_a} \cup_{b \in R_\alpha}b$. This redefinition of variables is achieved by the unitary circuit $U$. 
\section{Perturbative Derivation of Effective Hamiltonians for Product Codes} \label{app:pertforproduct}
In the main text, we present effective Hamiltonians that include both commuting $Z$-type code constraints and transverse-field $X$ terms. The $Z$-type terms alone suffice to define the classical constraint structure
associated with a cLDPC code. However, when discussing \emph{quantum phases of matter} and \emph{state preparation}, it is standard to include transverse fields, which select a trivial product-state phase at weak coupling and enable adiabatic interpolation into nontrivial phases.

Because these transverse-field terms do not commute with the strong interlayer coupling constraints, their effects are naturally treated
using degenerate perturbation theory in the strong-coupling limit. In this appendix, we provide a detailed perturbative derivation of the effective Hamiltonians quoted in the main text, including both $Z$- and $X$-type terms. We begin by reviewing the general perturbative framework \cite{bravyiSchriefferWolffTransformationQuantum2011} in Sec.~\ref{app:pertgeneral}, then derive the general effective Hamiltonians for the tensor product in Sec.~\ref{app:perttensor} and check product in Sec.~\ref{app:pertcheck}, and finally work out explicit examples for the 1D Ising model in Sec.~\ref{app:pertexamples}.

\subsection{General Perturbative Framework}\label{app:pertgeneral}
In both the tensor-product and check-product constructions, the total Hamiltonian
can be written as
\[
H = H_0 + V,
\]
where $H_0 := -\lambda \sum_{i,j} \Pi_{ij}$ is the strong interlayer coupling term, and 
\[
V = H_Z + H_X,
\]
with
\begin{align}
H_Z &= -\sum_{a,j} C^1_{a,j} - \sum_{i,b} C^2_{i,b},
\\
H_X &= -h_1 \sum_{i,j} \alpha_{ij}^x - h_2 \sum_{i,j} \beta_{ij}^x.
\end{align}
Let $P$ be the projector onto the ground space of $H_0$, and $Q=\openone-P$.

In the strong-coupling limit
$\lambda \gg \{1,h_1,h_2\}$,
the low energy effective Hamiltonian acting on $P\CH$ admits the standard
degenerate perturbation expansion
\begin{align}
H_{\mathrm{eff}}
=
P H_0 P
+ P V P
+ P V \frac{Q}{E_0-H_0} V P
+ \cdots ,
\label{eq:Heff_general}
\end{align}
where $E_0$ is the ground state energy of $H_0$.
We now apply this framework to the two product constructions.
\subsection{Tensor Product} \label{app:perttensor}
For the tensor product construction, the interlayer constraint is
\begin{align}
\Pi_{ij} := \alpha_{ij}^z \beta_{ij}^z.
\end{align}
The ground space of $H_0=-\lambda \sum_{i,j} \alpha_{ij}^z \beta_{ij}^z$ is characterized by the local constraint
$\alpha_{ij}^z = \beta_{ij}^z$ for all $(i,j)$.

Within this constrained subspace, we define emergent qubits
\begin{align}
\sigma_{ij}^z := P \alpha_{ij}^z P = P \beta_{ij}^z P,
\qquad
\sigma_{ij}^x := P(\alpha_{ij}^x \beta_{ij}^x)P.
\end{align}
Single-spin flips $\alpha_{ij}^x$ or $\beta_{ij}^x$ violate the constraint
and therefore satisfy $P\alpha_{ij}^xP=P\beta_{ij}^xP=0$.
However, the paired operator $\alpha_{ij}^x\beta_{ij}^x$ preserves the
constraint and acts nontrivially within $P\CH$.

The $Z$-check terms commute with $H_0$ and project directly:
\begin{align}
P H_Z P
=
-\sum_{a,j} \prod_{i\in \delta^1(a)} \sigma_{ij}^z
-\sum_{i,b} \prod_{j\in \delta^2(b)} \sigma_{ij}^z.
\end{align}
Because $P H_X P = 0$, the first nontrivial contribution from $H_X$ appears at
second order:
\begin{align}
H^{(2)}_{\rm eff}
=
P H_X \frac{Q}{E_0 - H_0} H_X P .
\end{align}
Acting on a constrained state, $\alpha_{ij}^x$ (or $\beta_{ij}^x$) creates a
single violated constraint with energy cost $\sim 2\lambda$.
A second flip on the same site with the opposite species restores the
constraint, yielding $\alpha_{ij}^x \beta_{ij}^x$.
Up to constants, this produces
\begin{align}
H^{(2)}_{\rm eff}
\;\sim\;
-\frac{h_1 h_2}{\lambda}
\sum_{i,j} \sigma_{ij}^x
\;+\; O(\lambda^{-2}) .
\label{eq:app_tensor_Heff2}
\end{align}

Collecting terms, the effective Hamiltonian in the constrained subspace is
\begin{align}
H_{\rm eff}
=
-\sum_{a,j}
\prod_{i\in \delta^1(a)} \sigma_{ij}^z
-\sum_{i,b}
\prod_{j\in \delta^2(b)} \sigma_{ij}^z
-\frac{h_1 h_2}{\lambda}
\sum_{i,j} \sigma_{ij}^x
\;+\; O(\lambda^{-2}) ,
\label{eq:app_tensor_final}
\end{align}
which matches \eqnref{eq:tpeff}.

\subsection{Check Product} \label{app:pertcheck}
For the check-product construction, the interlayer constraint is
\begin{align}
\Pi_{ij} := \alpha_{ij}^x \beta_{ij}^x .
\end{align}
The constrained ground state subspace of $H_0 = -\lambda \sum_{i,j} \alpha_{ij}^x \beta_{ij}^x$  is characterized by the constraint $\alpha_{ij}^x=\beta_{ij}^x$.

We define emergent variables
\begin{align}
\gamma_{ij}^x := P\alpha_{ij}^xP = P\beta_{ij}^xP,
\qquad
\gamma_{ij}^z := P(\alpha_{ij}^z \beta_{ij}^z)P.
\end{align}
Now $P H_X P$ contributes at first order because $\alpha_{ij}^x$ and $\beta_{ij}^x$ act identically
inside $P\CH$:
\begin{align}
P H_X P = -(h_1+h_2)\sum_{i,j}\gamma_{ij}^x.
\label{eq:app_check_PHXP}
\end{align}
In contrast, the $Z$-check terms in $H_Z$ generally do not preserve the constraint subspace at low order:
each factor $\alpha_{ij}^z$ or $\beta_{ij}^z$ flips $\alpha_{ij}^x$ or $\beta_{ij}^x$ and creates
constraint violations. Consequently, generating an effective term built from $\gamma^z$ requires a sequence
of virtual processes that returns to the constrained subspace.

Consider a pair of checks $a\in A$ and $b\in B$.
The operator
$\prod_{i\in\delta^1(a)} \alpha_{ij}^z$
flips the $XX$ constraints on all sites in
$\delta^1(a)\times\{j\}$,
and similarly
$\prod_{j\in\delta^2(b)} \beta_{ij}^z$
flips constraints on
$\{i\}\times\delta^2(b)$.
To return to the constrained subspace, one must apply a collection of such $Z$-operators so that every violated constraint is repaired.
The lowest nontrivial process that produces the paired operator
\[
\prod_{(i,j)\in \delta^1(a)\times\delta^2(b)}
\alpha_{ij}^z \beta_{ij}^z
\]
appears at perturbative order
\[
r = \lvert \delta^1(a) \rvert + \lvert \delta^2(b) \rvert .
\]
At that order, one obtains an effective interaction proportional to the
check-product stabilizer:
\begin{align}
H_{\rm eff}
=
-(h_1+h_2)\sum_{i,j}\gamma_{ij}^x
-
\sum_{a\in A}\sum_{b\in B}
\frac{
\kappa \,
}{
\lambda^{\lvert \delta^1(a) \rvert + \lvert \delta^2(b) \rvert}
}
\prod_{i\in\delta^1(a)}
\prod_{j\in\delta^2(b)}
\gamma_{ij}^z
\;+\; O(\lambda^{-(\lvert \delta^1(a) \rvert + \lvert \delta^2(b) \rvert+1)}) ,
\label{eq:app_check_final}
\end{align}
where $\kappa$ is a constant. This matches \eqnref{eq:cpeff}.
The emergent $Z$-terms are precisely the $Z$-checks of the check-product code
$\CC_1 \star \CC_2$, while $\sum_{i,j} \gamma_{ij}^x$ provides the
transverse-field drive.

\subsection{Examples}\label{app:pertexamples}
In this subsection we work out two canonical examples in detail:
(i) tensor product of two 1D Ising models, which yields a 2D Ising model, and
(ii) check product of two 1D Ising chains, which yields a 2D plaquette Ising model.
These examples illustrate explicitly how the perturbative mechanism described above reproduces the expected effective Hamiltonians and phases.
\subsubsection{Tensor product example: \texorpdfstring{$I_{1\mathrm{D}}\otimes I_{1\mathrm{D}} \rightarrow I_{2\mathrm{D}}$}{I1D tensor I1D to I2D}}
We take both $\CC_1$ and $\CC_2$ to be the one-dimensional Ising code, with Hamiltonians
\begin{align}
H_A^{(j)} &= - \sum_i \alpha_{ij}^z \alpha_{i+1,j}^z - h_1 \sum_i \alpha_{ij}^x, \\
H_B^{(i)} &= - \sum_j \beta_{ij}^z \beta_{i,j+1}^z - h_2 \sum_j \beta_{ij}^x.
\end{align}
We impose the strong $Z$-type coupling
\begin{align}
H_0 = -\lambda \sum_{i,j} \alpha_{ij}^z \beta_{ij}^z,
\qquad
\lambda \gg \{1,h_1,h_2\}.
\end{align}

Let $P$ denote the projector onto the ground space of $H_0$, characterized by
$\alpha_{ij}^z\beta_{ij}^z=+1$ for all $(i,j)$.
Within this constrained subspace we define emergent spins
\begin{align}
\sigma_{ij}^z := P\alpha_{ij}^z P = P\beta_{ij}^z P,
\qquad
\sigma_{ij}^x := P(\alpha_{ij}^x\beta_{ij}^x)P.
\end{align}

Projecting the $Z$-interactions gives
\begin{align}
P H_Z P
=
-\sum_{i,j} \sigma_{ij}^z \sigma_{i+1,j}^z
-\sum_{i,j} \sigma_{ij}^z \sigma_{i,j+1}^z,
\end{align}
which are precisely the horizontal and vertical Ising couplings on a square lattice.

Single-spin flips $\alpha_{ij}^x$ or $\beta_{ij}^x$ violate the $ZZ$ constraint and are
projected out at first order.
At second order, paired flips on the same site restore the constraint,
\begin{align}
H^{(2)}_{\rm eff}
\sim
-\frac{h_1 h_2}{\lambda}\sum_{i,j} \sigma_{ij}^x
\quad + \quad \text{const.}
\end{align}

Collecting terms,
\begin{align}
H_{\rm eff}
=
-\sum_{i,j} \sigma_{ij}^z \sigma_{i+1,j}^z
- \sum_{i,j} \sigma_{ij}^z \sigma_{i,j+1}^z
-\frac{h_1 h_2}{\lambda}\sum_{i,j} \sigma_{ij}^x
+ O(\lambda^{-2}),
\end{align}
which is the transverse-field Ising model in two dimensions. Applying the KW transformation discussed in the main text then maps this
to a $\mathbb{Z}_2$ topologically ordered phase.

\subsubsection{Check product example: \texorpdfstring{$I_{1\mathrm{D}}\star I_{1\mathrm{D}}\rightarrow PI_{2\mathrm{D}}$}{I1D star I1D to PI2D}}
We again start from two stacks of 1D Ising chains,
but now impose an $X$-type coupling,
\begin{align}
H_0 = -\lambda \sum_{i,j} \alpha_{ij}^x \beta_{ij}^x,
\qquad
\lambda \gg \{1, h_1,h_2\}.
\end{align}
The constrained subspace is characterized by $\alpha_{ij}^x=\beta_{ij}^x$.
We define emergent degrees of freedom
\begin{align}
\gamma_{ij}^x := P\alpha_{ij}^x P = P\beta_{ij}^x P,
\qquad
\gamma_{ij}^z := P(\alpha_{ij}^z\beta_{ij}^z)P.
\end{align}
Unlike the tensor product case, the transverse fields contribute at first order:
\begin{align}
P H_X P = -(h_1+h_2)\sum_{i,j} \gamma_{ij}^x.
\end{align}
The $Z$-interactions violate the $XX$ constraint and only re-enter the constrained
subspace through higher-order virtual processes.
For the 1D Ising checks, the lowest nontrivial contribution appears at fourth order and
produces plaquette interactions,
\begin{align}
H^{(4)}_{\rm eff}
\sim
-\frac{3}{128\,\lambda^4}
\sum_{i,j}
\gamma_{ij}^z \gamma_{i+1,j}^z \gamma_{i,j+1}^z \gamma_{i+1,j+1}^z
\quad + \quad \text{const.}
\end{align}
The resulting low-energy theory is
\begin{align}
H_{\rm eff}
=
-(h_1+h_2)\sum_{i,j} \gamma_{ij}^x
-\kappa \sum_{i,j}
\gamma_{ij}^z \gamma_{i+1,j}^z \gamma_{i,j+1}^z \gamma_{i+1,j+1}^z
+ O(\lambda^{-5}),
\end{align}
with $\kappa \sim 1/\lambda^4$.
This is the 2D plaquette Ising model in a transverse field. Thus, the check product construction produces subsystem-type constraints and, in appropriate
parameter regimes, exhibits fracton-like order.

\section{Cubic product as a \texorpdfstring{$(p=3,q=2)$}{(p=3,q=2)}-product}\label{app:pq}
Given cLDPC codes $\CC_1$, $\CC_2$, $\CC_3$, the cubic product is a $(p,q)=(3,2)$ instance of the general $(p,q)$-product construction. It can be constructed in two steps:
\begin{enumerate}
    \item Take the check product of every 2-element subset of $[3]=\{1,2,3\}$
    \begin{align*}
        \CC_{\{1,2\}}=\CC_1 \star \CC_2, \quad \CC_{\{1,3\}}=\CC_1 \star \CC_3, \quad \CC_{\{2,3\}}=\CC_2 \star \CC_3
    \end{align*}
    \item Take the tensor product of these three check-product codes
    \begin{align*}
        \CC_{(3,2)}=\CC_{\{1,2\}} \otimes \CC_{\{1,3\}} \otimes \CC_{\{2,3\}}
    \end{align*}
\end{enumerate}
Let $\CC_1$, $\CC_2$, $\CC_3$ have boundary maps $\delta^1$, $\delta^2$, $\delta^3$, sets of checks $A=\{C_a^1\}$, $B=\{C_b^2\}$, $C=\{C_c^3\}$, embedding Hilbert spaces $\CH_\alpha$, $\CH_\beta$, $\CH_\gamma$ and associated Hamiltonians
\begin{align}
    H_1&=- \sum_{a \in A}C^1_a - h_1 \sum_i \alpha_i^x,\\
    H_2&=- \sum_{b \in B}C_b^2 - h_2\sum_j \beta_j^x,\\
     H_3&=- \sum_{c \in C}C_c^3 - h_3\sum_k \epsilon_k^x,
\end{align}
where $C^1_a=\prod_{i \in \delta^1(a)}\alpha_i^z$, $C^2_b=\prod_{j \in \delta^2(b)}\beta_j^z$, and $C^3_c=\prod_{k \in \delta^3(c)}\epsilon_k^z$. 
The cubic product can be realized as an operator
\[
\mathsf{Cub}: \CH_\alpha^{\otimes n_2n_3} \otimes \CH_\beta^{\otimes n_1n_3} \otimes \CH_\epsilon^{\otimes n_1n_2} \to \CH_\sigma .
\]
The cubic product admits a physical realization via a coupled-layer construction that combines
the building blocks of the tensor and check products. We stack $n_2n_3$ copies of $\CH_\alpha$, $n_1n_3$ copies of $\CH_\beta$, $n_1n_2$ copies of $\CH_\epsilon$. Then, we introduce a coupling Hamiltonian
\begin{align}
    H_{\mathrm{cpl}}=-\lambda\sum_{i,j,k}(\alpha_{ijk}^x\beta_{ijk}^x+\alpha_{ijk}^x\epsilon_{ijk}^x+\beta_{ijk}^x\epsilon_{ijk}^x).
\end{align}
Adding this coupling term to the stacked Hamiltonians and taking the strong-coupling limit $\lambda\gg 1, h_1,h_2,h_3$, the effective Hamiltonian for the resulting cubic product code up to $\max(|\delta^1(a)|+|\delta^2(b)|,|\delta^1(a)|+|\delta^3(c)|,|\delta^2(b)|+|\delta^3(c)|)$-th order perturbation is
\begin{align}
    H_{\text{eff}}\simeq&- \frac{\kappa_1}{(-\lambda)^{|\delta^1(a)|+|\delta^2(b)|}} \sum_{k}\sum_{a \in A, b \in B} \prod_{i \in \delta^1(a),j \in \delta^2(b)}\sigma_{ijk}^z- \frac{\kappa_2}{(-\lambda)^{|\delta^1(a)|+|\delta^3(c)|}} \sum_{j}\sum_{a \in A, c \in C} \prod_{i \in \delta^1(a),k \in \delta^3(c)}\sigma_{ijk}^z\\&- \frac{\kappa_3}{(-\lambda)^{|\delta^2(b)|+|\delta^3(c)|}} \sum_{i}\sum_{b \in B, c \in C} \prod_{j \in \delta^2(b),k \in \delta^3(c)}\sigma_{ijk}^z-(h_1+h_2+h_3)\sum_{ijk}\sigma_{ijk}^x.
\end{align}
As a concrete example, taking $\CC_1=\CC_2=\CC_3=\CI_{1D}$ yields an effective Hamiltonian
equivalent to the 3D plaquette Ising model. This model exhibits planar subsystem symmetries and, in appropriate parameter regimes, realizes a fracton-like phase, which demonstrates how the $(p,q)$-product systematically generates higher-order constraint structures.

%\bibliographystyle{apsrev4-2}
%\bibliography{references}

\end{document}